\documentclass[11pt]{article}

\usepackage[utf8]{inputenc}
\usepackage[margin=2.5cm]{geometry}
\usepackage[T1]{fontenc}
\usepackage[english]{babel}
\usepackage{amsmath,amssymb,amsthm}
\usepackage[dvipsnames,table]{xcolor}
\usepackage[bibstyle=alphabetic, sorting=nyt, citestyle=alphabetic, giveninits, maxbibnames=10, url=false, isbn=false]{biblatex} 
\usepackage{mathtools, ifthen, xparse}
\usepackage[shortlabels]{enumitem}
\usepackage{apptools}
    \AtAppendix{\counterwithin{thm}{section}}
\usepackage{hang}
\usepackage{tikz,tikz-cd}
\usepackage{csquotes}\MakeOuterQuote"
\usepackage{hyperref}
\usepackage[capitalize]{cleveref}
    \hypersetup{pdftitle={The Schmidt rank for the commuting operator framework},pdfauthor={van Luijk, Schwonnek, Stottmeister, Werner}, verbose=false, colorlinks, citecolor=MidnightBlue, linkcolor=MidnightBlue,urlcolor=OliveGreen}

\newtheorem{thm}{Theorem}
\crefname{thm}{Theorem}{Theorems}
\newtheorem*{thm*}{Theorem}
\newtheorem{lem}[thm]{Lemma}
\newtheorem{cor}[thm]{Corollary}
\newtheorem{prop}[thm]{Proposition}

\newtheorem*{prob*}{Problem}

\newtheorem{defin}[thm]{Definition}

\crefname{defin}{Definition}{Definitions}
\theoremstyle{definition}
\newtheorem{exa}[thm]{Example}
\newtheorem{rem}[thm]{Remark}
\theoremstyle{remark}
\newtheorem*{ex*}{Example}
\newtheorem*{rem*}{Remark}

\theoremstyle{remark}

\theoremstyle{plain}
\newtheorem*{introthm}{Theorem} 
\newtheorem*{introdef}{Definition}

\newcommand*{\1}{\text{\usefont{U}{bbold}{m}{n}1}}

\newcommand\CC{\mathbb C}

\newcommand\eps\varepsilon

\renewcommand\H{\mathcal H}
\newcommand\K{\mathcal K}
\newcommand\V{\mathcal V}
\newcommand\NN{\mathbb N}
\def\placeholder{\,\cdot\,} 

\newcommand\RR{\mathbb R}
\renewcommand{\S}{{\mathcal S}}
\newcommand\restrictedto\upharpoonright
\newcommand{\states}{{\mathfrak S}}

\DeclareMathOperator{\rank}{rank}
\DeclareMathOperator\tr{Tr}

\newcommand\oo\infty
\newcommand\ox\otimes
\newcommand\ZZ{\mathbb Z}
\newcommand\QQ{\mathbb Q}

\NewDocumentCommand\TC{o}{{\IfNoValueTF{#1}{\mathcal{T}(\mathcal{H})}{\mathcal{T}(#1)}}}
\newcommand\boundedoperatorsymbol{{\mathcal B}}
\NewDocumentCommand\BO{mo}{ \IfNoValueTF{#2}{\boundedoperatorsymbol(#1)}{\boundedoperatorsymbol(#1,#2)}}

\def\up#1{^{(#1)}}

\let\mc\mathcal
\let\mf\mathfrak

\DeclareMathOperator\id{id}
\DeclareMathOperator\lin{span}

\let\O\relax

\let\lim\relax
\NewDocumentCommand\lim{o}{\IfNoValueTF{#1}{\mathop{\textup {lim}}}{\mathop{\textup{{$ #1 $}-lim}}}}

\DeclareFontFamily{U}{matha}{\hyphenchar\font45}
\DeclareFontShape{U}{matha}{m}{n}{ <-6> matha5 <6-7> matha6 <7-8> matha7 <8-9> matha8 <9-10> matha9 <10-12> matha10 <12-> matha12 }{}
\DeclareSymbolFont{matha}{U}{matha}{m}{n}
\DeclareFontFamily{U}{mathx}{\hyphenchar\font45}
\DeclareFontShape{U}{mathx}{m}{n}{ <-6> mathx5 <6-7> mathx6 <7-8> mathx7 <8-9> mathx8 <9-10> mathx9 <10-12> mathx10 <12-> mathx12 }{}
\DeclareSymbolFont{mathx}{U}{mathx}{m}{n}
\DeclareMathDelimiter{\vvvert} {0}{matha}{"7E}{mathx}{"17}%
\DeclarePairedDelimiterX{\normiii}[1]{\vvvert}{\vvvert} {\ifblank{#1}{\:\cdot\:}{#1}}

\DeclarePairedDelimiterX\norm[1]\lVert\rVert{\ifblank{#1}{\placeholder}{#1}}
\DeclarePairedDelimiter\abs\lvert\rvert

\DeclarePairedDelimiterX\ip[2]{\langle}{\rangle}{#1 , #2}
\DeclarePairedDelimiterX\dual[2]{\langle}{\rangle}{#1 , #2}

\DeclarePairedDelimiter\ket\vert\rangle
\DeclarePairedDelimiter\bra\langle\vert
\DeclarePairedDelimiterX\braket[2]\langle\rangle{ #1 \delimsize\vert #2}
\DeclarePairedDelimiterX\braAket[3]\langle\rangle{ #1 \delimsize\vert #2 \delimsize\vert #3 }
\DeclarePairedDelimiterX\ketbra[2]\vert\vert{ #1 \delimsize\rangle\delimsize\langle #2 }
\DeclarePairedDelimiterX\kettbra[1]\vert\vert{ #1 \delimsize\rangle\delimsize\langle #1 }

\providecommand\given{}  
\newcommand{\SetSymbol}[1][]{ \nonscript\ #1\vert \allowbreak \nonscript\ \mathopen{} } 
\DeclarePairedDelimiterX{\set}[1]\{\}{ \renewcommand\given{\SetSymbol[\delimsize]} #1 }

\newcommand\hide[1]{}

\def\A{{\mathcal A}} 
\def\B{{\mathcal B}}

\def\EE{{\mathbb E}}
\def\M{{\mathcal M}}
\def\MM{{\mathbb M}}
\def\N{{\mathcal N}}
\def\R{{\mathcal R}}

\def\B{{\mathcal B}}
\def\C{{\mathcal C}}
\def\K{{\mathcal K}}

\def\O{\mathcal O}

\def\V{\mathcal{V}}
\def\W{{\mathcal W}}
\newcommand\SR{\mathrm{SR}}
\def\Bar#1{\overline{#1}}
\def\op{{\oplus}}

\DeclareMathOperator{\ntr}{tr}

\title{The Schmidt rank for the commuting operator framework}
\author{Lauritz van Luijk, Ren\'e Schwonnek, \\ Alexander Stottmeister, and Reinhard F.\ Werner} 
\date{\normalsize Institut f\"ur Theoretische Physik, Leibniz Universi\"at Hannover, \\ Appelstraße 2, 30167 Hannover, Germany\\[2ex]\today}

\begin{document}

\maketitle

\begin{abstract}\noindent
In quantum information theory, the Schmidt rank is a fundamental measure for the entanglement dimension of a pure bipartite state.
Its natural definition uses the Schmidt decomposition of vectors on bipartite Hilbert spaces, which does not exist (or at least is not canonically given) if the observable algebras of the local systems are allowed to be general C*-algebras.
In this work, we generalize the Schmidt rank to the commuting operator framework where the joint system is not necessarily described by the minimal tensor product but by a general bipartite algebra.
We give algebraic and operational definitions for the Schmidt rank and show their equivalence.
We analyze bipartite states and compute the Schmidt rank in several examples:
The vacuum in quantum field theory, Araki-Woods-Powers states, as well as ground states and translation invariant states on spin chains which are viewed as bipartite systems for the left and right half chains.
We conclude with a list of open problems for the commuting operator framework.
\end{abstract}

\tableofcontents

\section{Introduction}\label{sec:intro}

Specifying the dimension of a concrete quantum system at hand is a task that arguably does not have a unique answer. There is an inherent ambiguity in deciding which possible degrees of freedom have to be modeled as quantum, which can be left out, and which can be regarded as part of an unspecified environment.
In quantum information theory, when formulated on Hilbert spaces, the Schmidt rank is usually used to describe an effective local dimension of a pure bipartite state. It connects to regarding entanglement as the relevant quantity and can intuitively be understood as the minimal number of local (quantum) degrees of freedom needed for implementing a desired quantum state. 

Entanglement, one of the most characteristic phenomena of quantum physics, can moreover be found in all kinds of quantum systems, including those finite-dimensional Hilbert spaces can not describe.  Here the mathematical construction used in the usual definition of the Schmidt rank does, however, no longer apply. Nevertheless, specifying an effective entanglement dimension is still a well-motivated and vital question. 
In the present work, we will explore paths for extending the concept of a Schmidt rank accordingly. Instead of Hilbert spaces, we will consider the more general algebraic formulation of quantum mechanics and introduce several, as we will show, equivalent algebraic and operational definitions of the Schmidt rank for this setting.  

As a start, let us recall the mathematical construction on which the usual definition of a Schmidt rank is based: 
Let $\mathcal{H}_A$ and $\mathcal{H}_B$ be separable Hilbert spaces and let $\ket{\Omega}\in\mathcal{H}_A\otimes \mathcal{H}_B$ be a vector in their Hilbert space tensor product. A basic result, attributed to Erhard Schmidt \cite{schmidt}, states that there exist subspaces $\mathcal{K}_A\subseteq \mathcal{H}_A$ and $\mathcal{K}_B\subseteq \mathcal{H}_B$ spanned by orthonormal bases $\{\ket{\Phi^A_i}\}$ and $\{\ket{\Phi^B_i}\}$, such that 
$\ket\Psi$ can be decomposed  as
\begin{align}\label{eq:usualschmidt}
\ket\Psi =\sum_{i=1}^k \lambda_i \ket{\Phi^A_i}\otimes\ket{\Phi^B_i}
\end{align}
with (up to reordering) unique coefficients $\lambda_i>0$.
\cref{eq:usualschmidt} is called Schmidt decomposition, the coefficients $\{\lambda_i\}$ are called Schmidt spectrum, and the number $k=\dim\mathcal{K}_A=\dim\mathcal{K}_B$ is called Schmidt rank. 

In the usual Hilbert space formulation of quantum mechanics, a pure quantum state of a composite system, consisting of two parties $A$ and $B$, is modeled by a unit vector $\ket\Psi$ as above.
Here the Schmidt decomposition is commonly interpreted as a reduction of $\ket\Psi$ to a  composition of physically relevant local subsystems, modeled by the subspace $\mathcal{K}_A\otimes \mathcal K_B$. 
If $k=1$, those subspaces are one dimensional, and $\ket\Psi$ can be written as a product of individual states on $A$ and $B$. Such a state shows no correlations and is called separable. 
In contrast, if $k>1$, the state $\Psi$ will be entangled. Accordingly, the number $k$ is also referred to as the entanglement dimension.   

For quantum technology, states with high Schmidt rank are regarded as a resource, and it is a technical goal to build devices that realize a $k$ as high as possible \cite{marcus,aubrun2022monogamy,ecker-entanglement-distribution,
zhu-high-dimensional-photonic-entanglement,skrzypczyk2015loss,marciniak2015unbounded,qu2022retrieving,miklin2022exponentially}. 
The strength of a quantum advantage can often be connected to scaling with $k$, for example, within the capacity of quantum communication channels, for quantum metrology, or quantum computing. Moreover, states with large $k$ allow for improving noise robustness of protocols and experiments, which is one of the central bottlenecks for any near-term quantum technology. 
Accordingly, there also has been much effort in developing techniques for certifying the Schmidt rank of a quantum system when only partial data is available. Along this research line, so-called Schmidt rank-$k$ steering recently became a focus \cite{designolle-highdim-steering, designolle2022robust, qu2022retrieving, sekatski2023unlimited}. Here, the task of certifying the presence of a Schmidt rank $k$ state is considered in situations where only the Hilbert space of one party, say $A$, is specified, whereas the system of the other party $B$ is left uncharacterized. 
In such a situation, the common method for modeling quantum systems without a specified Hilbert space is to take the more general perspective of algebraic quantum mechanics employing a description in terms of operator algebras \cite{Florian, landsman, redei2007quantum}. 
The original motivation for this work came from the primordial task of giving a mathematically consistent definition of the Schmidt rank in this situation. 
Our analysis goes, however, far beyond this example.

In general, there are many situations in which it is natural to model quantum systems by specifying their observable algebras $\A$. 
Typical examples range from device-independent cryptography\cite{primaatmaja2023security,tan2021computing,zhang2022device}, classical-quantum hybrid systems \cite{lars}, and statistical mechanics \cite{bratteli2} to non-local games \cite{lupini2020perfect} and quantum field theory \cite{haag1996lqp}. 
This algebraic formulation also offers a unified description of classical and quantum systems (as well as the hybrids mentioned earlier).
In this framework, states are modeled by normalized positive linear functionals $\omega:\A\to\CC$ on the observable algebra. 
Furthermore, the joint system of two separated parties $A$ and $B$ will be most generally modeled by some algebra $\mathcal{A}$ that contains two commuting subalgebras $\A_A$ and $\A_B$. 
In this setting, there is, at least a priori, no obvious extension of a Schmidt decomposition in terms of tensor products of Hilbert space vectors as in \eqref{eq:usualschmidt} and, hence, no direct generalization of the Schmidt rank. 
Here, underlying concepts have to be clarified, and definitions have to be made.  

Hurdles along this path stem from the fact that the r.h.s.\ of \eqref{eq:usualschmidt} does not necessarily have counterparts in an algebraic setting.  
At first glance, the algebraic and the Hilbert space view on pure states may seem to be interchangeable: Operators on a Hilbert space can be well regarded from a purely algebraic perspective, and abstract states can be conversely represented on concrete Hilbert space via the GNS construction. 
A clear distinction arises, however, when inspecting the two different definitions of bipartite systems.
Here the algebraic definition of bipartite systems via commuting operators is strictly more general than the Hilbert space definition via tensor products of local spaces. 
A prominent consequence of this culminated in what became known as Tsirelson's problem \cite{tsirelsons_prob, Junge2011} asking for observable differences between the two definitions on the level of linear correlations, which was answered affirmatively recently \cite{Ji2020}. 
In our case, a further hint to the challenges arising when generalizing the Schmidt rank can be seen when representing a pure state $\omega$ as a vector in a Hilbert space $\mathcal{H}_\omega$ via the GNS construction.
On $\mathcal{H}_\omega$, the observables of $A$ and $B$ act as commuting operators, but $\mathcal{H}_\omega$ will generally not admit a factorization into a tensor product of Hilbert spaces separating the respective subsystem.
In such a case, the concept of local spaces $\mathcal{K}_A$, $\mathcal{K}_B$ that form the overall system via a tensor product, and by this also a Schmidt decomposition as in \eqref{eq:usualschmidt}, does not exist. 
Coming from the other side also sets challenges because fixing local systems $\mathcal{A}_A$ and $\mathcal{A}_B$ does not lead to a unique algebra $\mathcal{A}$ for the joint system.
In contrast to the finite-dimensional case, there are several inequivalent ways of combining two given subsystems into a larger composite system. 
This choice is not only purely mathematical but rather an essential part of the physical model under consideration. The algebra $\mathcal{A}$ determines in which way two systems couple and by this also which physical interactions will be possible or not.    
As a consequence, an extension of the Schmidt rank demands a more refined view, and the path to take may strongly depend on the state under consideration.

In this work, we investigate three, as we will show, equivalent approaches for extending the Schmidt rank.  
Our first approach builds on algebraic properties of the GNS representation. The other approaches are operational in the sense that they extend the role of the Schmidt rank as the dimension of a minimal effective Hilbert space. This can be done in an entanglement-based picture and in a prepare-and-measure picture. The intuitive description of these approaches is placed in \cref{sec:summary} while the mathematical details are worked out in \cref{sec:Schmidt}.
Our definitions apply in a setting in which we only assume that two local algebras $\mathcal{A}_A$ and $\mathcal{A}_B$ are embedded in a larger algebra $\mathcal{A}$ as commuting subalgebras.
By this approach, inequivalent ways of coupling two systems are taken into account. More details on this and the connections to Tsirrelson's problem are discussed in \cref{sec:cof}. 

In \cref{sec:examples}, we discuss the following examples and applications: Gapped ground states and finitely correlated states on spin chains in the thermodynamic limit, Araki-Woods-Powers states, and bipartite states arising in quantum field theories from causally separated regions.
In particular, we show that the Schmidt rank of the ground state of the Heisenberg anti-ferromagnet is infinite, while the ground state of the AKLT model has a Schmidt rank of two, as one would expect from analyzing these models on finite-length chains.
In the last section \cref{sec:outlook}, we present some open problems.

\paragraph{Notation and conventions.}
Inner products are linear in the second entry.
All $C^*$-algebras are assumed to be unital, and the unit is denoted $1$, with an appropriate subscript for emphasis if necessary.
The GNS representation of a state $\omega$ on a $C^*$-algebra $\A$ is denoted $(\pi_\omega,\H_\omega,\Omega_\omega)$.
The unit interval of a $C^*$-algebra $\A$, i.e.\ the set of operators $x\in\A$ with $0\le x\le 1$, is denoted $[0,1]_\A$.
The algebra of bounded operators on a Hilbert space $\H$ is denoted $\B(\H)$, and the identity operator is $\1$.
If $\M\subset\B(\H)$ and $\V \subset\H$ then we denote by $[\M\V]$ the closed linear hull of the vectors $x\Psi$, $x\in\M$, $\Psi\in\V$, an we write $[\M\Psi]$ for $[\M\{\Psi\}]$.
If $X$ is a compact Hausdorff space, we denote the $C^*$-algebra of bounded continuous functions on $X$ by $C(X)$.
The standard basis of $\CC^n$ is denoted $\ket1,\ldots,\ket n$ and the algebra of $n\times n$ matrices is denoted $\MM_n$.

\section{Definition of the Schmidt rank and summary of main results}\label{sec:summary}

In this section, we explain our approach to the commuting operator framework, the definition of the Schmidt rank, and summarize our results.
In total, we find six equivalent definitions of the Schmidt rank (see \cref{thm:schmidt_rank}) from which we, at this point, only discuss the three central ones. 
We focus on conceptual ideas and give physical intuitions wherever we can.
A self-contained mathematical treatment is given in the subsequent sections.

\subsection{Bipartite algebras, bipartite states, and the commuting operator framework}\label{sec:summary_bip_alg}

In a correlation experiment, we have two physical systems with experimenters, conventionally named Alice and Bob, performing local measurements on a shared state.
Both parties are free to choose their measurements from their respective observable algebras $\A_A$ and $\A_B$.
According to the probabilistic principles of quantum theory \cite{Holevo2001}, the measurement statistics can be described by a bilinear functional $\omega_0:\A_A\times\A_B\to\CC$ such that $\omega(1_A,1_B)=1$ and $\omega_0(a,b)\ge0$ if $a\ge0$ and $b\ge0$.
The correlations can only be explained by quantum theory if there is a state on a larger system that contains Alice's and Bob's systems as subsystems and induces these correlations.
The larger system's observable algebra $\A$ contains $\A_A$ and $\A_B$ as commuting subalgebras.
The commutativity expresses that Alice and Bob's systems are \emph{kinematically independent} \cite{summers1990independence}.
Furthermore, there is no measurement apparatus accessible to both Alice and Bob, so their observable algebras have trivial intersection $\A_A\cap\A_B=\CC1$ in $\A$.%
\footnote{If one allows for an overlap of the observable algebras, it follows from commutativity that the intersection $\A_A\cap\A_B$ is a commutative subalgebra and, hence, corresponds to shared randomness accessible to both Alice and Bob. 
For two classical systems, this would even allow for $\A_A=\A_B=\A$.
However, it makes little sense to speak about correlation experiments if Alice and Bob have access to the same observable.}
That we may restrict $\A$ to the subalgebra generated by $\A_A$ and $\A_B$ leads us to: 

\begin{introdef}
    Let $\A_A$ and $\A_B$ be $C^*$-algebras. 
    A {\bf bipartite algebra} (for $\A_A$ and $\A_B$) is a $C^*$-algebra $\A$ with embeddings $\A_A\hookrightarrow\A$ and $\A_B\hookrightarrow\A$ such that $\A_A$ and $\A_B$ commute in $\A$, $\A_A\cap\A_B=\CC1$ and such that $\A$ is generated by $\A_A\cup\A_B$ as a $C^*$-algebra.
    A {\bf bipartite state} for $\A_A$ and $\A_B$ is a state $\omega$ on a bipartite algebra $\A$.
\end{introdef}

Therefore, the correlations $\omega_0$ can be explained by quantum theory if there is a bipartite state $\omega$ such that $\omega_0(a,b)=\omega(ab)$.
It turns out that a necessary and sufficient condition for this is that $\omega_0$ satisfies the following stronger version of positivity
\begin{equation}\label{eq:q_constraint}
    \sum_{ij}\omega_0(a_i^*a_j,b_i^*b_j)\ge 0\qquad\forall a_1,\ldots,a_n\in\A_A,\, b_1,\ldots,b_n\in\A_B.
\end{equation}
This property is called the quantum constraint.
Details can be found in \cref{sec:bip_alg}.

The bipartite algebra describes the joint system and should be regarded as a part of the mathematical model as it determines how the local systems are coupled to each other, e.g.\ by determining which global interactions are possible.
For some bipartite algebras, Alice and Bob's systems are not \emph{statistically independent} in the sense that Alice and Bob are constrained in their ability to perform local operations such as state preparations \cite{summers1990independence}.
Statistical independence is formally defined in \cref{sec:bip_alg}, where it is shown that statistical independence holds if and only if the bipartite algebra is a $C^*$-tensor product (see \cref{thm:independence}).

In general, a representation $\pi$ of a bipartite algebra $\A$ on a Hilbert space $\H$ will represent Alice and Bob's observables only as commuting operators $\pi(a)\pi(b)=\pi(b)\pi(a)$ with no guarantee for a tensor splitting $\pi(ab)=\pi_A(a)\ox\pi_B(b)$.
Bipartite states induce, in general, proper commuting operator framework correlations, i.e.\ correlations that cannot even be approximated with a tensor splitting of Hilbert spaces.

Even though Alice and Bob have full control of their local algebras, there is no method for detecting the concrete form of the bipartite algebra itself through a correlation experiment in which all measurements are performed on a fixed bipartite state $\omega$. 
Therefore, we call bipartite states $\omega_1$ and $\omega_2$ \emph{correlation-equivalent} if they induce the same correlations $\omega_1(ab)=\omega_2(ab)$. We call a property $X$ a \emph{correlation invariant} if it is invariant under local unitaries and assigns the same value to correlation-equivalent states.
Details on this equivalence relation can be found in \cref{sec:bip_alg}, where we show that many important properties such as purity, Haag-duality, or the Schmidt rank are indeed correlation invariants.
Another interesting correlation invariant is a certain subfactor inclusion induced by the GNS representation of a pure bipartite.

\subsection{An algebraic definition}\label{sec:alg}

We are now able to follow different paths for extending the concept of a Schmidt rank. 
We will refer to our first approach as the \textit{algebraic definition}.
Recall that a major hurdle for a naive extension of \cref{eq:usualschmidt} is that there are states for which there is no tensor factorization of the GNS Hilbert space that separates Alice and Bob's observables.
A straightforward idea for circumventing this issue is to take the ad hoc ansatz of initially only considering states for which a factorization exists. On those, a Schmidt rank definition similar to the one derived from \eqref{eq:usualschmidt} can be applied. As we will see, it makes sense to elevate this idea to a proper definition by conventionally assigning an infinite Schmidt rank to all other states.

Via the following theorem, the slightly makeshift factorization property of a GNS space can be connected to the concrete categorization of von Neumann algebras into different types. We say that a state $\varphi$ is of type I, II, or III if the von Neumann algebra generated in its GNS representation is of type I, II, or III. 
We will denote the marginals of a bipartite state $\omega$, i.e.\ the restrictions to the algebras $\A_A$ and $\A_B$, by $\omega_A$ and $\omega_B$, respectively.
For pure bipartite states, i.e.\ pure states on some bipartite algebra, there is an intimate connection between the von Neumann type of the marginals and the factorization of the GNS space:

\begin{introthm}\label{thm:pure_tame_states_summary}
    Let $\omega$ be a pure bipartite state. The following are equivalent
    \begin{enumerate}[(1)]
   		\item there are irreducible representations $\pi_j:\A_j\to\B(\H_j)$, $j=A,B$, and a vector $\Omega\in\H_A\ox\H_B$ so that $\omega(ab)=\ip\Omega{\pi_A(a)\ox\pi_B(b)\Omega}$ for all $(a,b)\in\A_A\times\A_B$,
        \item either $\omega_A$ or $\omega_B$ is a type I state,
        \item both $\omega_A$ and $\omega_B$ are type I states,
    \end{enumerate} 
\end{introthm}

We call a pure bipartite state {\bf tame} if it satisfies these equivalent properties and {\bf wild} otherwise. 
Tame bipartite pure states are precisely those whose correlations can be reproduced with tensor products of Hilbert spaces and vector states (we extend this notion to mixed states in \cref{sec:bip_alg}).
For tame bipartite pure states, the representations $\pi_j$ and the vector $\Omega$ are unique up to unitary equivalence and constitute the GNS representation.
An essential consequence of the theorem is that tameness can be decided from the knowledge of a single marginal.
Equipped with this, we can state the algebraic definition of the Schmidt rank:

\begin{introdef}[Algebraic definition]\label{def:alg}
    If $\omega$ is tame, the Schmidt rank of $\omega$ is the Schmidt rank of the GNS vector $\Omega_\omega$ with respect to the tensor splitting $\H_\omega=\H_A\ox\H_B$ of the GNS space.
    If $\omega$ is wild, the Schmidt rank is infinite.
\end{introdef}

We briefly comment on the structural properties of pure bipartite states before we discuss the Schmidt rank further.
For bipartite states $\omega$ one gets commuting von Neumann algebras $\M_A = \pi_\omega(\A_A)''$ and $\M_B=\pi_\omega(\A_B)''$ acting on the GNS space $\H_\omega$. 
The objects $(\M_A,\M_B,\H_\omega,\Omega_\omega)$ can be shown to be a correlation invariant. 
For a pure state, $\M_A$ and $\M_B$ are even factors, i.e.\ von Neumann algebras with trivial center, and jointly generate $\B(\H_\omega)$.
These factors are tame (resp.\ wild) if and only if $\omega$ is a tame (resp.\ wild) pure bipartite state.
For pure bipartite states, $\M_A\subset \M_B'$ is an {\it irreducible subfactor inclusion}, i.e.\ an inclusion of factors with trivial relative commutant $\M_A'\cap\M_B'=\CC\1$, and conversely, every irreducible subfactor inclusion arises in this way.
It follows from this that purity, Haag-duality, and the Jones index $[\M_A:\M_B']$ are correlation invariants for pure states.

Introducing the class of tame states and an according concept of Schmidt rank, as above, gives a nice ordering within  the set of pure bipartite states. 
It can be summarized as follows: 
\[
    \Big\{\!\!\begin{array}{c}\text{\small product}\\[-0.1cm] \text{\small states}\end{array}\!\!\Big\} =\SR_{1} \subset \SR_{\le2}\subset\ldots\subset \states_{tame} \subset \states_{min}\subset \states_{max}.
\] 
All product states and all states with a (due to this definition) finite Schmidt rank are tame.  As in  the Hilbert space case, we can sort states into sets of Schmidt rank not bigger than $k$. This gives a chain interpolating between product states and states with infinite entanglement. 
The set $ \states_{tame}$ of tame states itself  is included in the set $\states_{min}$ of states on the minimal $C^*$-tensor product, which itself is a subset of the set  $\states_{max}$ of states corresponding to the maximal tensor product. Remarkably, all those inclusions are strict in general. With an appropriate extension of tameness to non-pure bipartite states, we will show that tame states are a dense convex subset of $\states_{min}$.

When coming from a Hilbert space-centered perspective, one should note that tame bipartite states are precisely those that can be represented as normal states on a tensor product of Hilbert spaces. Here the distinction to arbitrary states on a minimal tensor product lies in the fact that the latter ones could also be singular, i.e.\ proper elements of $\mathcal{B}\left(\mathcal{H}_A\otimes\mathcal{H}_B\right)^*$ which can not be identified with elements of $\mathcal{T}\hspace{-0.1cm}\left(\mathcal{H}_A\otimes\mathcal{H}_B\right)$. 
From a hands-on perspective, demanding a state to be tame can be very practical: even if we are confronted with the most bizarre algebras $\mathcal{A}_A$ and $\mathcal{A}_B$,  the broad range of methods and calculations developed in the Hilbert space formalism can still be applied without taking too much special care. However, a striking operational justification for this definition is, at this point, not apparent. This lack is fixed by the two other definitions that we will consider.

\subsection{An operational definition via minimal compressions}
In our second approach, we define the Schmidt rank in terms of local compressions (see Fig.~\ref{fig:compression}).
In a correlation experiment, all accessible information on a quantum state is captured by considering its behavior on all possible pairs of local measurements.
For Alice and Bob at working with a bipartite state $\omega$ on some big bipartite algebra, one can ask whether it is possible to emulate/reproduce the behavior of $\omega$ on local measurements with some clever protocol that only requires a bipartite quantum state on a small Hilbert space as a resource.  
We will refer to such a protocol as compression. The minimal dimension into which Alice and Bob can compress their observables without losing any information will give us a definition for the Schmidt rank.  

At this point, we will omit any consideration of classical communication and model compression protocols by a pair of local quantum channels.
To implement an emulation of $\omega$, we will grant Alice and Bob joint access to a quantum system $\mathcal{B}(\mathcal{K})$ corresponding to some Hilbert space $\mathcal{K}$. Their actions during such a protocol are described by local quantum channels. Since we  do not want to assume a tensor product structure at this stage, the locality of Alice's and Bob's operations is modeled by demanding that the actions of these channels commute. 
It is convenient to formulate the following definition in the Heisenberg picture.

\begin{figure}[htpb!]
\centering
\includegraphics[width=0.65\textwidth]{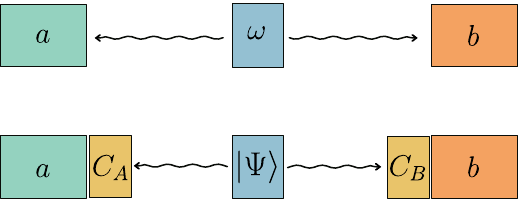}
\caption{Operational interpretation of the Schmidt rank through source emulation. A state $\omega$ has Schmidt rank $k$ if it can be prepared by preparing a vector state $\ket\Psi$ with Schmidt rank $k$ and performing local operations.}
\label{fig:compression}
\end{figure}

A {\it compression} with respect to a bipartite state $\omega$ is a collection  consisting of a pair of unital completely positive maps $C_j:\A_j\to \B(\K)$, $j=A,B$, with commuting ranges and a unit vector $\Psi\in\K$ such that 
\begin{equation}
    \omega(ab) = \ip\Psi{C_A(a)C_B(b)\Psi}\qquad\forall (a,b)\in\A_A\times\A_B.
\end{equation}
This concept of compressions is investigated in detail in \cref{sec:Schmidt}.
Since we are working in the Heisenberg picture, these compressions are compressions of measurements.
On the level of states, the channels $C_A$ and $C_B$ map the bipartite state $\psi=\ip\Psi{(\placeholder)\Psi}$ to $\omega$ by only applying local operations.

\begin{introdef}[Compression definition]
    The Schmidt rank of a pure bipartite state $\omega$ is 
    \begin{equation}
        \SR(\omega) := \sqrt{\min_{(C_A,C_B,\K,\Psi)} \dim(\K)\,}
    \end{equation}
    where the minimum is over all compressions with respect to $\omega$.
\end{introdef}

We will prove that the minimum dimension is indeed a square number so that the Schmidt rank is guaranteed to be an integer (unless it's infinite).
In quantum communication tasks, this definition is relevant for the class of entanglement-based protocols. It is well-known that each such protocol has a counterpart in what is called a `prepare and measure' scenario, which will be discussed subsequently. 
The equivalence between the scenarios is commonly established using the Schmidt decomposition. 
However, in the absence of this tool, it is reasonable to consider the next definition.

\subsection{An operational definition for the prepare and measure scenario}

As a third approach, we consider an encoding-decoding scenario. In a prepare and measure protocol, one party, say, Alice, prepares different quantum states $\{\omega^a\}$ and sends them to Bob, who then applies a measurement of his choice. 
This class of protocols is connected to bipartite states by the source-replacement scheme, an essential mathematical tool with wide use in quantum Cryptanalysis \cite{jie, pirandola}. 
In a virtual protocol, Alice's preparation of states $\{\omega^a\}$ is replaced by granting her access to a bipartite state $\omega$. By measuring an observable $a$ on one subsystem, say $\mathcal{A}_A$ she will create a conditional state 
 \begin{align}
 \omega^a=\omega(a (\placeholder))
 \end{align}
on the other subsystem $\mathcal{A}_B$. This state is then transferred to Bob. Here we can again ask for clever protocols that reduce the communication effort from Alice to Bob.
Such a protocol (see \cref{fig:factor}) will consist of an encoding that maps $\omega^a$ to a state on a small Hilbert space followed by a decoding that maps this state back to a state on $\mathcal{A}_B$. We can then take the smallest Hilbert space dimension for which such a protocol exists as a base for defining a Schmidt rank. 

\begin{figure}[htp]
\centering
\includegraphics[width=0.45\textwidth]{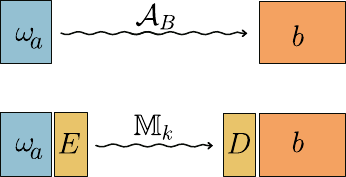}
\caption{Operational interpretation through encoding and decoding. Alice prepares a state $\omega^a$ 
by measuring an observable $a$ on one side of $\omega$. This state is then sent to Bob.
The smallest Hilbert space dimension $k$ into which this transmission can be encoded and decoded characterizes an effective dimension. We take this as a definition for the Schmidt rank.}
\label{fig:factor}
\end{figure}

It is equivalent to also think of such an encoding-decoding protocol as a map that takes Alice's observable $a$, used for state preparation, to a state on Bob's system, by passing it through an operator on the small Hilbert space.   
From this perspective, the existence of faithful encoders and decoders can conveniently be expressed in terms of factorization of the completely positive map $\Gamma_\omega: \A_A\ni a\mapsto \omega^a\in\A_B^*$.

\begin{introdef}[Factorization definition]\label{def:factor}
    Let $\omega$ be a pure state on a bipartite algebra.
    The Schmidt rank of $\omega$ is the smallest number $k\in\NN$ so that $\Gamma_\omega$ factorizes through $\MM_k$, i.e.\ so that there are completely positive maps $\alpha:\A_A\to\MM_k$ and $\beta:\MM_k\to\A_B^*$ so that the following diagram commutes:
    \begin{equation}
        \begin{tikzcd}
            \A_A \arrow{rr}{\Gamma_\omega} \arrow{dr}{\alpha} &&\A_B^*\\
            & \MM_k\arrow{ur}{\beta} &
        \end{tikzcd}
    \end{equation}
    If no such number exists, the Schmidt rank is defined to be infinite.
\end{introdef}

Without loss of generality, one may require the completely positive maps to satisfy (a) $\alpha$ maps $[0,1]_{\A_A}$ to $\{\rho\ge0:\tr\rho\le1\}$ and $\alpha(1)$ is a density operator and (b) $\beta$ is a quantum channel, i.e.\ takes density operators to states on $\A_B$.
This can then be interpreted as an encoding-decoding protocol: $\alpha$ realizes encodes $\omega_a$ into the (subnormalized) state $\alpha(a)$ on $\CC^k$ which is decoded by the quantum channel $\beta$ returning $\beta(\alpha(a)) = \omega^a$.

\subsection{Properties}

In the following, we list properties and results connected to the Schmidt rank in the commuting operator framework.
\begin{labeledlist}{l}
\item[\it Equivalence of definitions.]
    The three definitions of the Schmidt rank presented above are equivalent. The equivalence is proven throughout \cref{sec:proof}. 
    Indeed we could have extended this list further to, in total, six equivalent definitions of the Schmidt rank (see \cref{thm:schmidt_rank}). At this point, we will, however, list some of those other equivalent definitions among this list of properties. 

\item[\it Consistency with existing definitions.]
    Most importantly, the definitions of the Schmidt rank reduce to the normal definition in the case of finite-dimensional quantum systems.
    I.e.\ if $\A_A=\MM_{d_A}$ and $\A_B=\MM_{d_B}$, there is a unique bipartite algebra given by $\MM_{d_A}\otimes\MM_{d_B}=\MM_{d_A\cdot d_B}$;
    every pure state $\psi$ on this bipartite algebra is implemented by a vector $\Psi\in\CC^{d_A}\ox\CC^{d_B}$ and our definition assigns to $\psi$ the vector Schmidt rank of $\Psi$.
    This can be best seen  from the perspective offered by the algebraic formulation. 
    The GNS representation of $\psi$ is just the standard representation of $\MM_{d_A}\ox\MM_{d_B}$  on $\CC^{d_A}\ox\CC^{d_B}$ and the GNS vector is the vector $\Psi\in\CC^{d_A}\ox\CC^{d_B}$ implementing $\psi$.

\item[\it Tensor product splitting of the minimal compression.] 
    The definition via  compressions asks for compressibility into an arbitrary target Hilbert space $\mathcal{K}$ with the constraint that Alice and Bob's compressed observables commute. As it turns out, we have that the Hilbert space corresponding to the minimal compression will always admit a tensor product  splitting between Alice's and Bob's observables if the Schmidt rank is finite. In other words, we equivalently could also have asked for compressions as a pair of unital completely positive maps $C_j:\A_j\to\B(\K_j)$ and a unit vector $\Psi\in\K_A\ox\K_B$ such that $\omega(ab)=\ip\Psi{C_A(a)\ox C_B(b)\Psi}$.
    Details are explained in \cref{sec:Schmidt}.

\item[\it One system finite dimensional.]
    In the case that one system, say Alice's, is finite-dimensional, we can map the system of Bob to an effective system with the same dimension. 
    In detail, we have that if system $A$ is an $n$ dimensional quantum system, i.e.\ $\A_A = \MM_n$, then the Schmidt rank of a bipartite pure state $\omega$ is the rank of the density operator $\rho$ defined by $\tr[\rho a] = \omega(a\ox1_{B})$.
    In particular, the Schmidt rank of all bipartite pure states is bounded by $n$, regardless of the algebra $\A_B$.
    We have the following canonical form:
    Let $\Psi\in\CC^n\ox\CC^r$ be the canonical purification of $\rho$ and let $k=\mathrm{rank}(\rho)$.
    There is a unique surjective unital completely positive map $T_B : \A_B \to \MM_k$ so that $\omega(ab) = \ip\Psi{a\otimes T_B(b)\Psi}$.
    This is proved in \cref{sec:properties}.

\item[\it Monotonicity under local operations.]
    Since entanglement can only decrease under local operations, any sensible measure of entanglement must decrease under local operations.
    A local operation $T$ is a pair of unital completely positive maps $T_A:\A_A\to\B_A$ and $T_B:\A_B\to\B_B$ between the local algebras of two bipartite systems.
    From the definition in terms of minimal compressions, it is immediate that the Schmidt rank is indeed monotonically decreasing, i.e.\ if $\omega$ is mapped to $\varphi$ by a local operation, then $\SR(\omega)\ge\SR(\varphi)$.
    As a consequence, we obtain that all pure bipartite states with infinite distillable entanglement have infinite Schmidt rank. Proofs of these statements can be found in \cref{sec:properties}.

\item[\it Schmidt rank via the rank of the marginals.] 
    We generalize the concept of the rank of a density operator to general states on $C^*$-algebras.
    The rank of a state is the minimum number of pure states required to write it as a convex combination.
    With this generalization, we prove that the Schmidt rank is equal to the rank of the marginal states of a bipartite pure state $\omega$, i.e.\ $\SR(\omega)=\rank(\omega_A)=\rank(\omega_B)$. For details see \cref{sec:Schmidt}.

\item[\it Tame bipartite states.]
    We introduce the notion of tame bipartite states. These are states that can be obtained through shared randomness from states that are emulable with density operators and a tensor product splitting between Alice and Bob's observables.
    All separable states are tame, and all tame states can be represented on the minimal $C^*$-tensor product. In fact, they form a $w^*$-dense convex subset.
    This is examined in detail in \cref{sec:tame_wild}.

\item[\it Connection to Tsirelson's problem.]
    The hierarchy of correlation bodies studied in the context of Tsirelson's problem is a special case of the classes of bipartite pure states given by the states with bounded Schmidt rank, tame states, and the state spaces of the minimal and maximal $C^*$-tensor product. This is investigated in detail in \cref{sec:Tsirelsons_prob}.

\item[\it Systems admitting only tame bipartite states.]
    Consider a fixed system $A$ with observable algebra $\A_A$.
    Irrespective of $\A_B$, all bipartite pure states are tame if and only if $\A_A$ is a type I $C^*$-algebra.
    This is a well-studied class of $C^*$-algebras which are particularly well-behaved. This is proved in \cref{sec:bip_alg}.
    
\item[\it Computability.]
    The algebraic definition of the Schmidt rank allows for an explicit computation.
    For example, we obtain explicit values for the Schmidt rank of the ground state of the Heisenberg antiferromagnet model and the generalized AKLT model, where states on infinite spin chains are viewed as bipartite states for the left and right sides.
\end{labeledlist}

\section{The commuting operator framework}\label{sec:cof}

\subsection{Bipartite algebras}\label{sec:bip_alg}

Throughout, $A$ and $B$ are physical systems described by observable algebras $\A_A$ and $\A_B$.
If our systems are, however, described by more general $C^*$-algebras, there is no unique description of the joint system.
We define:

\begin{defin}\label{def:bipartite_alg}
    Let $C^*$-algebras $\A_A$, $\A_B$ describing systems $A$ and $B$ be given.
    A {\bf bipartite algebra} is a $C^*$-algebra $\A$ together with $^*$-embeddings $\A_A\hookrightarrow \A$ and $\A_B\hookrightarrow\A$ such that $\A_A$ and $\A_B$ commute in $\A$, have trivial intersection $\A_A\cap\A_B=\CC$ and generate  $\A$ as a $C^*$-algebra.
    A {\bf bipartite state} for $\A_A$ and $\A_B$ is a state $\omega$ on a bipartite algebra $\A$.
    If we want to emphasize the bipartite algebra we denote the bipartite state by $(\omega,\A)$.
\end{defin}

The most important class of bipartite algebras are the $C^*$-tensor products.
A $C^*$-tensor products are constructed by completing the algebraic tensor product $\A_A\odot\A_B$ with respect to a $C^*$-norm%
\footnote{The algebraic tensor product $\A_A\odot\A_B$ has the structure of a $^*$-algebra. A $C^*$-norm on a $^*$-algebra is a norm $\norm{}_\beta$ on which is $^*$-invariant, sub-multiplicative and satisfies the $C^*$-property $\norm{xx^*}_\beta = \norm{x}_\beta^2$.} $\norm{}_\beta$ and is denoted $\A_A\ox_\beta\A_B$ \cite[Sec.~IV.4]{takesaki1}.
By identifying $\A_A$ with $\A_A\ox1_B$ and $\A_B$ with $1_A\ox\A_B$, a $C^*$-tensor product indeed is a bipartite algebra.
There are two canonical $C^*$-norms, called the minimal and the maximal $C^*$-norm, because all other $C^*$-norms are bounded from below by the minimal and from above by the maximal norm.
The minimal norm $\norm{}_{min}$ is defined by taking arbitrary faithful representations $\A_j\subset\B(\H_j)$, $j=A,B$, and using the operator norm for the induced representation of $\A_A\odot\A_B$ on $\H_A\ox\H_B$. The resulting norm is independent of the chosen representations.
The maximal norm $\norm{}_{max}$ is defined by maximizing the operator norm over all commuting operator representations:
\[
    \norm[\big]{\sum_i a_i\ox b_i}_{max} = \sup_{\pi_A,\pi_B} \,\norm[\big]{\sum_i \pi_A(a_i)\pi_B(b_i)}
\]
where the supremum is over all pairs of representations $\pi_j$ on the same Hilbert space $\H_{AB}$ such that $\pi_A(\A_A)$ and $\pi_B(\A_B)$ commute.
The resulting algebra is the universal $C^*$-algebra generated by commuting operators $a\in\A_A$ and $b\in\A_B$ (see below).

\hide{The bipartite algebra usually used in algebraic formulations of quantum mechanics is the minimal $C^*$-tensor product $\A_{min}=\A_A\ox_{min}\A_B$. It can be constructed from faithful representations $\A_j\hookrightarrow\B(\H_j)$ as the $C^*$-subalgebra of $\B(\H_A\ox\H_B)$ generated by operators $a\ox b$ with $(a,b)\in\A_A\times\A_B$.In general, a $C^*$-tensor product $\A_\beta=\A_A\ox_\beta\A_B$ is a bipartite algebra that is constructed as the closure of the algebraic tensor product $\A_A\odot\A_B$ with respect to a $C^*$-cross-norm\footnote{A $C^*$-cross norm is a norm $\norm{}_\beta$ on the algebraic tensor product $\A_A\odot\A_B$ which is $^*$-invariant, submultiplicative and satisfies the $C^*$-property $\norm{xx^*}_\beta = \norm{x}_\beta^2$ as well as the cross-norm property $\norm{a\ox b}_\beta =\norm a\norm b$ (these conditions contain some redundancy \cite[Sec.~IV.4]{takesaki1}).} $\norm{}_\beta$. There is also a maximal $C^*$-tensor product $\A_{max}=\A_A\ox_{max}\A_B$ which has the universal property that for each $C^*$-algebra $\B$ and each pair of completely positive maps $\phi_A:\A_j\to\B$ with commuting ranges, there is a map $\phi:\A_{max}\to\B$ so that $\phi(a\ox b)=\phi_A(a)\phi_B(b)$ \cite[Prop.~IV.4.23]{takesaki1}.For all other $C^*$-tensor products $\A_\beta$, one has $\norm{x}_{min}\le\norm{x}_{\beta}\le\norm{x}_{max}$ for all $x\in\A_A\odot\A_B$.We will see that the $C^*$-tensor products are exactly those bipartite algebras for which the local systems are \emph{statistically independent}.}

\begin{defin}\label{def:local_op}
    Let $\A_A,\A_B$ and $\B_A,\B_B$ be two pairs of $C^*$-algebras and fix bipartite algebras $\A$ and $\B$.
    A {\bf local operation} is a unital completely positive map $T:\A\to\B$ such that $T(\A_j)\subset\B_j$, $j=A,B$.
\end{defin}

For general bipartite algebras $\A$ and $\B$, there is no guarantee that for two unital completely positive maps $T_j:\A_j\to\B_j$, there is a local operation $T:\A\to\B$ with $T|_{\A_j}=T_j$.
The maximal $C^*$-tensor product is the unique bipartite algebra

\begin{lem}[Universal property of $\A_{max}$]\label{thm:univ_prop}
    The maximal $C^*$-tensor is the only bipartite algebra that satisfies the following property: For every pair of unital completely positive maps $T_A:\A_A\to\B_A$ and $T_B:\A_B\to\B_B$ and every bipartite algebra $\B$, there is a local operation $T:\A_{max}\to\B$ so that $T|_{\A_j}=T_j$.
\end{lem}
\begin{proof}
    By \cite[Prop.~IV.4.23]{takesaki1}, the property holds if $\B=\B_{max}$. The general result follows by composition with the canonical homomorphism $\phi:\B_{max}\to\B$. To see uniqueness take $\B_j=\A_j$ and $T_j=\id_{\A_j}$.
\end{proof}

Bipartite states $\omega$ such that $\omega(ab)=\omega_A(a)\omega_B(b)$ for all $(a,b)\in\A_A\times\A_B$ are \emph{product states}.
If we regard $\CC$ as a bipartite system, then product states are local operations
The existence of sufficiently many product states is called statistical independence in \cite{summers1990independence}:

\begin{defin}
    Let $\A$ be a fixed bipartite algebra. $\A_A$ and $\A_B$ are {\bf statistically independent} if for every pair of states $\omega_A$ on $\A_A$ and $\omega_B$ on $\A_B$ there exists a product state $\omega$ on $\A$ whose marginals are $\omega_A$ and $\omega_B$.
\end{defin}

If statistical independence does not hold, Alice and Bob are not free to perform state preparations.
As a consequence of a result by Roos \cite{roos1970} (see \cite{summers1990independence,florig1997} for extensions), we get that statistical independence holds if and only if the bipartite algebra is a $C^*$-tensor product:

\begin{thm}\label{thm:independence}
    Let $\A$ be a bipartite algebra for $\A_A$ and $\A_B$. The following are equivalent
    \begin{enumerate}[(a)]
        \item\label{it:strong_statistical_indep} $\A_A$ and $\A_B$ are statistically independent.
        \item\label{it:statistical_indep} For all pairs of states $\omega_A$ and $\omega_B$ on $\A_A$ and $\A_B$ there is a state $\omega$ on $\A$ such that $\omega|_{\A_j}=\omega_j$ and $\omega|_{\A_B}=\omega_B$. 
        \item\label{it:operational_indep} All pairs of unital completely positive maps $T_j:\A_j\to\B_j$ can be combined into a local operation $T:\A\to\B_{min}$ such that $T|_{\A_j}=T_j$, $j=A,B$. Here $\B_{min}=\B_A\ox_{min}\B_B$ is the minimal $C^*$-tensor product.
        \item\label{it:algebraic_indep} For all $0\ne a\in\A_A$ and $0\ne b\in\A_B$ it holds that $ab\ne0$ in $\A$.
        \item\label{it:cstar_TP} $\A$ is a $C^*$-tensor product of $\A_A$ and $\A_B$, i.e.\ $\A\cong\A_A\ox_\beta\A_B$ for a $C^*$-norm $\norm{}_\beta$ on $\A_A\odot\A_B$.
    \end{enumerate} 
\end{thm}

\begin{proof}
    \ref{it:statistical_indep} $\Leftrightarrow$ \ref{it:algebraic_indep} was proved in \cite{roos1970}, \ref{it:cstar_TP} $\Rightarrow$ \ref{it:strong_statistical_indep} $\Rightarrow$ \ref{it:statistical_indep} and is clear.
    \ref{it:cstar_TP} $\Rightarrow$ \ref{it:operational_indep} can be seen by factoring the local operation $T$ through $\A=\A_{min}$ \cite[Prop.~IV.4.23]{takesaki1}.
    \ref{it:operational_indep} $\Rightarrow$ \ref{it:strong_statistical_indep} holds because product states are local operations into $\CC\ox_{min}\CC=\CC$.
    The direction \ref{it:algebraic_indep} $\Rightarrow$ \ref{it:cstar_TP} also follows from
    \cite{roos1970} where it is shown that \ref{it:algebraic_indep} implies that $\varphi:\sum a_i\ox b_i\mapsto \sum a_ib_i$ is an algebraic isomorphism from $\A_A\odot\A_B$ onto $\A_A\cdot\A_B\subset\A$.
    Therefore, $\norm{}_\beta=\norm{\varphi(\placeholder)}_\A$ inherits the properties of a $C^*$-norm from the norm on $\A$.
\end{proof}

\begin{cor}
    If either $\A_A$ or $\A_B$ is a nuclear simple $C^*$-algebra, there is a unique bipartite algebra $\A=\A_A\ox_{min}\A_B$.
    In particular, this holds one of the systems is a finite-dimensional quantum system, i.e.\ has observable algebra $\MM_n$.
\end{cor}

\begin{proof}
    By nuclearity we have $\A_{min}=\A_{max}$.
    It is proved in \cite{Ranjana2017} that the assumptions imply that all ideals of $\A_{min}$ are of the form $\mc J_A\ox\mc J_B$ for ideals $\mc J_j\subset\A_j$.
    If $\A$ is a bipartite algebra, then the kernel of the canonical homomorphism $\phi:\A_{max} \to \A$ is an ideal and, hence of product form.
    Since the embeddings $\A_j\hookrightarrow\A$ are isometric, $\mc J_j =\A_j$ follows, and we get $\A=\A_{max}= \A_{min}$.
\end{proof}

This means that we are guaranteed to have statistical independence whenever one of the parties has a finite-dimensional quantum system!

\begin{prop}\label{thm:bipartite_algebras}
    \begin{enumerate}[(1)]
        \item 
            For every, bipartite algebra $\A$ there is a unique surjective $^*$-homo\-morphism $\phi:\A_{max}\to\A$ so that $\phi(a\ox b)=ab$.
        \item\label{it:GNS_pullback}
            Let $\omega$ be a state on a bipartite algebra $\A$.
            Define a state $\omega_{max}=\omega\circ\phi$ on $\A_{max}$, then the GNS representation of $\omega_{max}$ can be constructed via $\H_{\omega_{max}}=\H_{\omega}$, $\pi_{\omega_{max}} = \pi_{\omega}\circ\phi$ and $\Omega_{\omega_{max}}=\Omega_{\omega}$.
        \item
            If $\tilde\A$ is a $C^*$-algebra that contains $\A_A$ and $\A_B$ as commuting subalgebras such that $\A_A\cap\A_B=\CC1$, then the subalgebra $\A\subset\tilde\A$ generated by $\A_A$ and $\A_B$ is a bipartite algebra.
    \end{enumerate} 
\end{prop}

\begin{proof}
    The first item follows from the universal property of the maximal $C^*$-tensor product \cite[Prop.~IV.4.23]{takesaki1}.
    The second item is evident from the uniqueness of the GNS representation.
    The last item is obvious.
\end{proof}

We consider some examples of bipartite algebras.

\begin{exa}\label{exa:bipartite_algebras}
\begin{enumerate}[(1)]
    \item 
        For our first example look at classical systems.
        Consider commutative $C^*$-algebras $\A_j=C(X_j)$ for compact Hausdorff spaces $X_A$ and $X_B$. 
        Then $\A = C(X)$, $X=X_A\times X_B$ is the unique $C^*$-tensor product where the embeddings $\A_j$ are defined by $f_j(x_A,x_B)=f(x_j)$, $f_j\in\A_j$, $j=A,B$.
    \item 
        In bosonic systems with one-particle space $\mf h$, we may take the observable algebra to be the CCR algebra $\mathrm{CCR}(\mf h)$.
        If $\mathfrak{h}_A$ and $\mathfrak{h}_B$ are one-particle Hilbert spaces, then the observable algebra of the joint system is $\mathrm{CCR}(\mf h_A\oplus\mf h_B)$.
        This natural bipartite algebra is isomorphic with the minimal $C^*$-tensor product of the local observable algebras $\mathrm{CCR}(\mf h_A\oplus\mf h_B)\cong \mathrm{CCR}(\mf h_A)\ox_{min}\mathrm{CCR}(\mf h_B)$.
    \item 
        Fermionic systems with one-particle space $\mathfrak{h}$ are described by the CAR algebra $\mathrm{CAR}(\mathfrak{h})$.
        The CAR algebra carries a natural grading and only elements of even parity make for valid physical observables.
        They form a unital subalgebra which we denote $\mathrm{CAR}_+(\mf h)$.\footnote{The CAR algebra $\mathrm{CAR}(\mf h)$ is the unital $C^*$-algebra generated by operators $a(\Psi)$, $\Psi\in\mf h$ such that $a(\lambda\Psi)=\Bar\lambda a(\Psi)$, $\{a(\Psi), a(\Phi)\}=0$ and $\{a(\Psi), a(\Phi)^*\}=\ip\Psi\Phi1$ for all $\Psi,\Phi\in\mf h$ and $\lambda\in\CC$, where $\{x,y\}=xy+yx$ is the anti-commutator. The elements of even parity are those that can be written as even polynomials in the generators $a(\Psi)$.}
        
        We want to describe correlations between two fermionic systems with one-particle spaces $\mf h_A$ and $\mf h_B$.
        The full CAR algebra of the joint system $\mathrm{CAR}(\mf h_A\oplus\mf h_B)$ is a \emph{graded} tensor product of $\mathrm{CAR}(\mf h_A)$ and $\mathrm{CAR}(\mf h_B)$.
        This means that the embeddings of the two local CAR algebras satisfy a graded version of commutativity in the full CAR algebra (i.e., some items commute and some items anti-commute).
        However, the local observable algebras $\A_j = \mathrm{CAR}_+(\mf h_j)$ of even-parity elements do commute. The natural bipartite algebra is the subalgebra $\A$ generated by $\A_A$ and $\A_B$, and this bipartite algebra is isomorphic with $\A_A\ox_{min}\A_B$.
        This bipartite $\A$ algebra is strictly contained in the observable algebra $\mathrm{CAR}_+(\mf h_A\op\mf h_B)$. 
        For example, $\A$ does not contain elements of the form $a(\psi_A\oplus0)a(0\oplus\psi_B)$ for $\psi_j\in\mf h_j$. Even though such elements have even parity and are seemingly product observables, they make no sense in a correlation experiment as the local operators admit no interpretation as local observables.
    \item\label{it:duck}
        Finally, we give an example where statistical independence does not hold.
        Consider again $\A_j=C(X_j)$ and $\A=C(X)$ but this time pick a proper subset $X\subset X_A\times X_B$ with full projections $\mathrm{pr}_j X = X_j$.
        Again the embedding $\A_j\hookrightarrow\A$ is defined by $f_j(x_A,x_B)=f_j(x_j)$.
        In this case, the systems $A$ and $B$ are not statistically independent since there is no product state whose marginals are $\delta_{x_A}$ and $\delta_{x_B}$ if $(x_A,x_B)\in X_A\times X_B\setminus X$.
        See \cref{fig:duck} for an illustration.
\end{enumerate} 
\end{exa}

\begin{figure}[htp]
    \begin{center}
        \includegraphics[scale=.25]{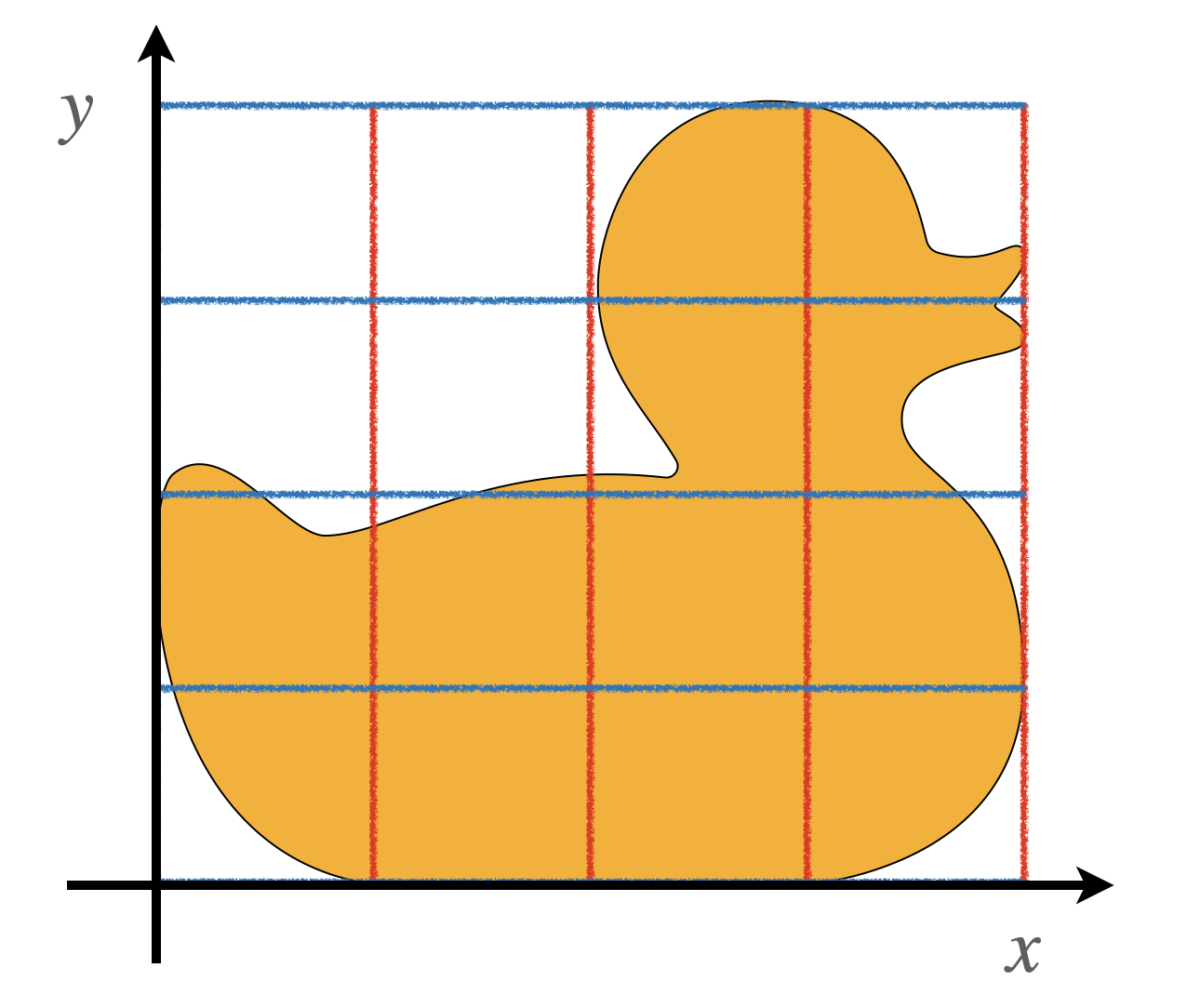}
    \end{center}
    \vspace{-22pt}
    \caption{Visualization of \cref{exa:bipartite_algebras}.\ref{it:duck}: A bipartite algebra for which the systems $A$ and $B$ are not statistically independent. The bipartite algebra is $\A=C(X)$ where $X$ is the yellow-colored region, i.e.\ the duck. $\A_A$ (resp.\ $\A_B$) are the subalgebras of functions that only depend on the $x$-variable (resp.\ $y$-variable).} 
    \label{fig:duck}
\end{figure}

\begin{defin}\label{def:corr_equiv}
    If we want to emphasize the algebra, we will denote a bipartite state by $(\omega,\A)$.
    Two bipartite states $\omega\up1$ and $\omega\up2$ are {\bf correlation-equivalent}, denoted $(\omega\up1,\A\up1)\sim(\omega\up2,\A\up2)$ (or $\omega\up1\sim\omega\up2$), if
    \begin{equation}\label{eq:equiv_rel}
        \omega\up1(ab)=\omega\up2(ab) \qquad\forall (a,b)\in\A_A\times\A_B.
    \end{equation}
\end{defin}

Correlation-equivalence is an equivalence relation. It means that $\omega\up1$ and $\omega\up2$ cannot be distinguished in pure correlation experiments, i.e.\ by local measurements only. 

\begin{prop}\label{thm:bipartite_states}
    \begin{enumerate}[(1)]
        \item\label{it:at_most_one}
            There is at most one correlation-equivalent bipartite state per bipartite algebra, i.e.\ $(\omega\up1,\A)\sim(\omega\up2,\A)$ implies $\omega\up1=\omega\up2$.
        \item\label{it:rep_on_Amax}
            Every equivalence class of bipartite states has a unique representative $\omega_{max}$ whose bipartite algebra is $\A_{max}$. It can be constructed from any representative $\omega$ as $\omega_{max}=\omega\circ\phi$ where $\phi:\A_{max}\to\A$ is the canonical $^*$-homomorphism.
            In particular, $(\omega\up1,\A\up1)\sim(\omega\up2,\A\up2)$ holds if and only if $\omega\up1\circ\phi_1=\omega\up2\circ\phi_2$.
        \item\label{it:GNS_equiv}
            If $(\omega_1,\A\up 1)\sim (\omega\up2,\A\up2)$, then their GNS representations are related by 
            \begin{equation}\label{eq:GNS_equiv}
                \H_{\omega\up1}=\H_{\omega\up2},\quad \pi_{\omega\up1}(ab)=\pi_{\omega\up2}(ab),\quad \Omega_{\omega\up1}=\Omega_{\omega\up2}.
            \end{equation}
    \end{enumerate} 
\end{prop}
\begin{proof}
    \ref{it:at_most_one}: This holds because $\lin(\A_A\cdot\A_B)$ is dense in every bipartite algebra $\A$.

    \ref{it:rep_on_Amax}: Uniqueness of the representative on $\A_{max}$ holds because of the first item, and existence follows from the construction using $\phi$.

    \ref{it:GNS_equiv}: This follows from \cref{thm:bipartite_algebras}.\ref{it:GNS_pullback} because $\omega\up1\circ\phi=\omega\up2\circ\phi$ by \ref{it:rep_on_Amax}.
\end{proof}

As a consequence, we can now show the claim from \cref{sec:summary_bip_alg} that correlations can be explained by quantum theory if and only if they satisfy the quantum constraint \eqref{eq:q_constraint}.

\begin{cor}
    The equation 
    \begin{equation}\label{eq:bipartite_states}
        \omega_0(a,b) = \omega(ab) = \omega_{max}(a\ox b) \qquad\forall(a,b)\in\A_A\times\A_B.
    \end{equation}
    determines a bijection between
    \begin{enumerate}[(i)]
        \item bilinear function $\omega_0:\A_A\times\A_B\to\CC$ with $\omega_0(1_A,1_B)=1$ and $\omega(a,b)\ge0$ if $a\ge0$ and $b\ge0$ that satisfy the quantum constraint \eqref{eq:q_constraint}.
       \item equivalence classes $[(\omega,\A)]$ of bipartite states with respect to correlation-equivalence,
       \item states $\omega_{max}$ on the maximal $C^*$-tensor product $\A_{max}$.
\end{enumerate}
\end{cor}
\begin{proof}
    The bijection between $\omega_0$ and $\omega_{max}$ is well-known, see for example \cite[Prop.~II.9.3.4]{blackadar2006oa} or \cite{lance1982tensor}. The bijection with $\omega_{max}$ was proved in the previous item.
\end{proof}

For a bipartite state $\omega$, we consider its GNS representation $(\H_\omega,\pi_\omega,\Omega_\omega)$.
We define von Neumann algebras $\M_j = \pi_\omega(\A_j)''\subset \B(\H_\omega)$.
Since $\A_A$ and $\A_B$ commute, so do $\M_A$ and $\M_B$, i.e.\ $\M_A\subset \M_B'$.
The data $(\M_A,\M_B,\H,\Omega)$ describes all correlations as $\omega(ab)=\ip\Omega{\pi(a)\pi(b)\Omega}$ for all $(a,b)\in\A_A\times\A_B$.

\begin{defin}\label{def:haag}
    \begin{enumerate}[(1)]
        \item 
            A pure state $\omega$ on a bipartite algebra $\A$ satisfies {\bf Haag-duality} if $\M_A = \M_B'$.
        \item 
            A state $\omega$ on a bipartite algebra satisfies the {\bf split property}, if there is a type I factor $\N\subset\B(\H)$ such that $\M_A\subset \N \subset \M_B'$ or, equivalently, if $\M_A\vee\M_B \cong \M_A\Bar\ox\M_B$, where "$\Bar\ox$" denotes the von Neumann tensor product \cite[Ch.~IV]{takesaki1}.
    \end{enumerate} 
\end{defin}

The split property implies that $\M_A$ and $\M_B$ can be separated by a tensor product splitting of $\H$ (but not necessarily in a unique way).
Both Haag-duality and the split property play important roles in algebraic quantum field theory \cite{Doplicher1984} and are intimately linked with the existence of (normal) product states \cite{buchholz1974product}.
We observe that being a product state $\omega$ has an important structural consequence for the von Neumann algebra $\pi_{\omega}(\A)''$:

\begin{lem}\label{lem:productstate}
    Let $\omega$ be a product state on a bipartite algebra $\A$. 
    Then its GNS representation can be constructed as $\H_\omega = \H_{\omega_A}\ox\H_{\omega_B}$, $\pi_\omega(ab)=\pi_{\omega_A}(a)\ox\pi_{\omega_B}(b)$ and $\Omega_\omega=\Omega_{\omega_A}\ox\Omega_{\omega_B}$.
    In particular, this implies $\pi_{\omega}(\A)''\cong\pi_{\omega_{A}}(\A_{A})''\!\ \Bar\ox\!\ \pi_{\omega_{B}}(\A_{B})''$.
\end{lem}
\begin{proof}
    This is a special case of \cref{thm:bipartite_states}.\ref{it:GNS_equiv}.
\end{proof}

\begin{defin}\label{def:correlation_inv}
    A property defined for all bipartite correlations will be called a {\bf correlation invariant} if it is constant on equivalence classes with respect to Correlation-equivalence and if it is invariant under local unitaries, i.e.\ assigns the same value to states $\omega$ and $u_Au_B\omega u_B^*u_A^* =\omega(u_Au_B\placeholder u_B^*u_A^*)$ for all pairs of unitaries $u_j\in\A_j$, $j=A,B$.
\end{defin}

\begin{thm}\label{thm:correlation_invariants}
    \begin{enumerate}[(1)]
        \item\label{it:reduced_GNS}
            The reduced GNS representation $(\pi_\omega|_{\A_A}, \pi_\omega|_{\A_B}, \H_\omega, \Omega_\omega)$ is a correlation invariant up to unitary equivalence. 
            In particular, $(\M_A,\M_B,\H_\omega,\Omega_\omega)$ is a correlation invariant.
        \item\label{it:purity}
            Being pure is a correlation invariant.
        \item\label{it:split_prop}
            The split property is a correlation invariant.
        \item\label{it:Haag-duality} 
            Haag-duality is a correlation invariant for pure bipartite states.
    \end{enumerate} 
\end{thm}

\begin{proof}
    \ref{it:reduced_GNS}:
    That these representations only change up to unitary equivalence if $\omega$ is replaced by $u_Au_B\omega u_B^*u_A^*$ is clear.
    The independence of the bipartite algebra is a special case of \cref{thm:bipartite_states}.\ref{it:GNS_equiv}.

    \ref{it:purity}:
    The purity of a state $\omega$ on a bipartite algebra $\A$ is equivalent to the irreducibility of the GNS representation $\pi_\omega$.
    The irreducibility holds if and only if $\pi_\omega(\A)'' = \M_A\vee\M_B$ is equal to $\B(\H_\omega)$.
    Therefore, the claim follows from \cref{it:reduced_GNS} because $\M_A,\M_B\subset \B(\H_\omega)$ is a correlation invariant.

    Similarly, \ref{it:split_prop} and \ref{it:Haag-duality} follow from \ref{it:reduced_GNS} because they only depend on the way that $\M_A$ and $\M_B$ act on $\H_\omega$.
\end{proof}

There is an intimate connection between \emph{pure} bipartite states and irreducible subfactor inclusions. 
An inclusion of factors $\R \subset \S$ is called irreducible if the relative commutant is trivial, i.e.\ if $\R'\cap \S=\CC1$.
This is summarized in the following:

\begin{prop}\label{thm:pure_states_and_subfactors}
\begin{enumerate}[(1)]
    \item 
        If $\omega$ is a pure bipartite state, $\M_A$ and $\M_B$ are factors acting on $\H_\omega$ and $\M_A \subset \M_B'$ is an irreducible subfactor inclusion.
        If $\A_A$ and $\A_B$ are separable, $\H_\omega$ is separable.
    \item 
        For every irreducible subfactor inclusion $\R\subset \S$ acting on a (separable) Hilbert space $\K$ there are (separable) $C^*$-algebras $\A_A$ and $\A_B$, and a pure bipartite state $\omega$ so that $\K=\H_\omega$, $\M_A=\R$ and $\M_B=\S'$.
\end{enumerate} 
\end{prop}

\begin{proof}
    The fact that $\M_A$ and $\M_B$ are factors follows as they jointly generate $\B(\H)$ and commute with each other so that the center of either algebra is contained in the center of $\B(\H)$ which is trivial.
    A similar argument shows $\M_A'\cap \M_B'=\CC\1$. 
    For the second item, pick (separable) $\sigma$-weakly dense $C^*$-subalgebras $\A_A\subset \R$ and $\A_B\subset \S'$ and define $\A\subset\B(\K)$ as the $C^*$-algebra generated by $\A_A$ and $\A_B$.
    Now pick any vector $\Omega\in \K$ and define $\omega$ as the corresponding vector state on $\A$.
    By construction $\A''=\B(\K)$, so that the identity is an irreducible representation of $\A$.
    In particular, $\Omega$ is a cyclic vector so that the GNS representation of $\omega$ is the identity, and $\omega$ is pure since the GNS representation is irreducible.
    By construction we also have $\M_A = \R$ and $\M_B = \S'$.
\end{proof}

As a consequence of \cref{thm:correlation_invariants}, we get that this subfactor inclusion $\M_A\subset\M_B'$ is a correlation invariant for pure bipartite states.
In particular, the Jones index $[\M_A:\M_B']$ of the subfactor inclusion, which has a physical interpretation in quantum field theory, is a correlation invariant (see \cite{kawahigashi} for an introduction to subfactor theory and the Jones index).

\begin{rem}
    The requirement that a bipartite algebra $\A$ is generated by the local algebras in norm makes sense from a $C^*$-algebraic point of view and certainly covers many examples.
    It does, however, not cover natural von Neumann algebraic examples such as $\A_j=\B(\H_j)$ and $\A=\B(\H_A\ox\H_B)$ if  both Hilbert spaces are infinite-dimensional.
    The solution here is obvious; one should take weak topologies into account when dealing with von Neumann algebras. 
    But there is also an alternative solution that bypasses the category of von Neumann algebras by making the allowed bipartite algebras state dependent:
    Define a {\bf bipartite system} to be a $C^*$-algebra $\A$ together with a state $\omega$ and $^*$-embeddings $\A_j\hookrightarrow\A$ with commuting ranges so that $\A_A\cap\A_B=\CC1$ and so that the $C^*$-algebra generated by $\A_A$ and $\A_B$ is dense in the topology generated by the seminorms $x\mapsto \omega(x^*x)$ and $x\mapsto \omega(xx^*)$.
    For such a bipartite system, one still obtains a canonical $^*$-homomorphism $\phi:\A_{max}\to\A$, which is, however, no longer surjective. Instead, its range is dense in the weak topology induced by $\omega$.
    This suffices to prove that the equivalence relation and notion of correlation invariants generalizes to bipartite systems so that \cref{thm:correlation_invariants} remains true.
\end{rem}

\subsection{Tame and wild bipartite states}\label{sec:tame_wild}

The maximal $C^*$-tensor product $\A_{max}=\A_A\ox_{max}\A_B$ is the unique bipartite algebra on which all bipartite states can be represented.
We denote its state space by $\states_{max}$.
The state space $\states(\A)$ of every other bipartite algebra can be regarded as a closed subset of $\states_{max}$ \cite{lance1982tensor}.\footnote{The embedding $\states(\A)\hookrightarrow\states_{max}$ is the dual map of the canonical surjective $^*$-homomorphism $\phi:\A_{max}\to\A$.}
The most important instance of this embedding is that state space $\states_{min}$ of the minimal $C^*$-tensor product is a closed subset of $\states_{max}$.
That a bipartite state is correlation-equivalent to a state in $\states_{min}$ can be regarded as a regularity condition which is, however, not very strict, as the following result shows:

\begin{lem}\label{thm:min_states}
    Let $\omega$ be a bipartite state. The following are equivalent:
    \begin{enumerate}[(a)]
        \item\label{it:factor_min} 
            $\omega$ is correlation-equivalent to a state on the minimal $C^*$-tensor product.
        \item\label{it:singular_state}
            There exist representations $\pi_j:\A_j\to\B(\H_j)$ and a possibly singular state $\eta$ on $\B(\H_A\ox\H_B)$ such that 
            \begin{equation}\label{eq:with_singular_state}
                \omega(ab)=\eta(\pi_A(a)\ox\pi_B(b))\qquad\forall(a,b)\in\A_A\times\A_B.
            \end{equation}
    \end{enumerate} 
\end{lem}

\begin{proof}
    \ref{it:singular_state} $\Rightarrow$ \ref{it:factor_min} is clear since the product representation factorizes through the minimal $C^*$-tensor product \cite[Prop.~IV.4.23]{takesaki1}.
    \ref{it:factor_min} $\Rightarrow$ \ref{it:singular_state}:
    Let $\A_j\subset\B(\H_j)$, $j=A,B$, be faithful representations.
    Then $C^*(\A_A\ox\1\cup\1\ox\A_B)\subset\B(\H_A\ox\H_B)$ is a representation of the minimal $C^*$-tensor product \cite[Prop.~IV.4.23]{takesaki1}.
    Define $\eta$ to be any Hahn-Banach state extension of the state $\omega$ on the subalgebra $\A_{min}$ to $\B(\H_A\ox\H_B)$.
\end{proof}

We want to define a notion for bipartite states that captures whether the correlations can be simulated using shared randomness and "ordinary quantum mechanics" only.
It is suggestive to just ask for the state $\eta$ in \cref{eq:with_singular_state} to be implemented by a density operator, i.e.\ that we can write $\omega$ as
\begin{equation}\label{eq:ordinaryQM}
    \omega(ab) = \tr[\rho(\pi_A(a)\ox \pi_B(b))]\qquad\forall(a,b)\in\A_A\times\A_B.
\end{equation}
This, however, also excludes certain bipartite states on classical systems.%
\footnote{Take for example $\A_j=C([0,1])$. If the bipartite state $\omega$ is the uniform probability distribution on the diagonal of the unit square $[0,1]^{\times 2}$, then the product of the marginals is the Lebsgue measure. Evidently, the diagonal is a null set, and $\omega$ is not absolutely continuous.}
The reason is that not all shared randomness can be cast into a bipartite Hilbert space setting.
To account for this, we allow for an arbitrary classical system shared by Alice and Bob:

\begin{defin}\label{def:tame}
    A bipartite state $\omega$ is {\bf tame} if there is a probability space $(X,\mu)$, a $w^*$-Borel measurable map $X\ni x\mapsto \omega_x\in\states_{max}$ such that each $\omega_x$ has a representation of the form \eqref{eq:ordinaryQM} for which 
     \begin{equation}\label{eq:mixed_tame}
         \omega(ab) = \int_X \omega_x(ab)\,d\mu(x)\qquad\forall(a,b)\in\A_A\times\A_B.
    \end{equation}
    All other bipartite states will be called {\bf wild}.
\end{defin}

It is now clear that all classical bipartite states are tame: If $\A_j=C(X_j)$, $j=A, B$, and bipartite state $\omega$ corresponds to a probability measure $\mu$ on $X=X_A\times X_B$, then $\omega(f)=\int_X f(x,y)\,d\mu(x,y) =\int_X \delta_x\ox \delta_y(f)\,d\mu(x,y)$.
Furthermore, we have:

\begin{lem}\label{thm:mixed_tame}
    \begin{enumerate}[(1)]
        \item\label{it:mixed_tame_min} Every tame bipartite state is correlation-equivalent to a state on the minimal $C^*$-tensor product.
        \item\label{it:mixed_tame_conv} The set $\states_{tame}$ of tame bipartite states on $\A_{min}$ is a convex dense subset of $\states_{min}$.
    \end{enumerate} 
\end{lem}
\begin{proof}
    \ref{it:mixed_tame_min}: States of the form $\tr[\rho(\pi_{A}(a)\ox\pi_{B}(b))]$ are correlation-equivalent to states on the minimal $C^*$-tensor product.
    By definition, tame states are elements of the $w^*$-closed convex hull of such states and, hence, also correlation-equivalent to states on the minimal $C^*$-tensor product.

    \ref{it:mixed_tame_conv}: Let $\omega\up1$ and $\omega\up2$ be tame bipartite states, $0<p<1$, and $\omega=p\omega\up1+(1-p)\omega\up2$.
    Let $(X\up i,\mu\up i)$ and $X\up i\ni x\mapsto\omega_{x}\up i$ be as in \cref{def:tame} with respect to $\omega\up i$.
    Then $X = X\up1 \cup X\up2$ is a probability space with $\mu=p\mu\up1+(1-p)\mu\up2$. 
    If we define a preparation $x\mapsto\omega_x$ piece-wise such that it is equal to  $\omega\up i_x$ on  $X\up i$ then we get $\omega(ab)=\int_X \omega_x\,d\mu(x)$.
\end{proof}

The main theorem of this section is a characterization of the tame property for pure states.
To state it, we recall some notions:
A state $\nu$ on a $C^*$-algebra $\B$ is called a \emph{factor state} if its GNS representation $\pi_\nu :\B \to \B(\H_\nu)$ is such that $\pi_\nu(\B)''$ is a factor.
It is called a \emph{type I (II, III) state} if $\pi_\nu(\B)''$ is a type I (II, III) von Neumann algebra.
The fact that every von Neumann algebra is a direct sum of type I, a type II, and a type III algebra implies that every state $\nu$ has a unique convex decomposition into a type I, a type II, and a type III state.
For factor states, it follows that they are either of type I, II, or III.
For example, every pure state is a type I factor state, and every state on $\MM_n$ is type I.
Two states $\rho$ and $\nu$ are \emph{quasi-equivalent}, if the GNS representations $\pi_\rho$ and $\pi_\nu$ can be intertwined with a normal $^*$-isomorphism $\phi:\pi_\rho(\B)''\to\pi_\sigma(\B)''$ in the sense that $\pi_\sigma=\pi_\rho\circ\phi$.
Quasi-equivalence is a rather loose notion of equivalence in terms of the physical properties of the state, e.g.\ \emph{all} states on $\MM_n$ are quasi-equivalent.
One can also understand quasi-equivalence as unitary equivalence up to multiplicity \cite[Sec.~2.4.4]{bratteli1}.

\begin{thm}\label{thm:pure_tame_states}
    Let $\omega$ be a pure state on a bipartite algebra $\A$. The following are equivalent:
    \begin{enumerate}[(a)]
        \item\label{it:tame}
            $\omega$ is tame.
        \item\label{it:one_typeI}
            Either $\omega_A$ or $\omega_B$ is a type I state.
        \item \label{it:both_typeI}
            Both $\omega_A$ and $\omega_B$ are type I states.
        \item\label{it:normal_form}  
            There are irreducible representations $\pi_j:\A_j\to\B(\H_j)$, $j=A,B$, and a vector $\Omega\in\H_A\ox\H_B$ so that $\ip\Omega{\pi_A(a)\ox\pi_B(b)\Omega}=\omega(ab)$ for all $(a,b)\in\A_A\times\A_B$.
        \item\label{it:quasi_equiv}
            $\omega$ is quasi-equivalent to a product state.
    \end{enumerate}
    If these properties hold, the representations $\pi_j$ and the vector $\Omega$ are unique up to unitary equivalence.
\end{thm}

In view of \cref{def:alg} of the Schmidt rank, we can regard tame bipartite pure states as those states whose Schmidt rank is at most countable infinite. 
Wild pure states, however, do not even allow a tensor product splitting separating Alice's and Bob's observables relative to the bipartite state, which prohibits any form of Schmidt decomposition.
Before we give the proof of this theorem, we look at some of its consequences.

\begin{cor}\label{thm:tame_GNS}
    Let $\omega$ be a tame bipartite pure state.
    \begin{enumerate}[(1)]
        \item\label{it:tame_GNS}
            The GNS space $\H_\omega$ has a unique tensor product decomposition $\H_\omega=\H_A\ox\H_B$ such that $\pi_\omega(\A_A)$ acts trivially on $\H_B$ while $\pi_\omega(\A_B)$ acts trivially on $\H_A$.
            The equation $\pi_\omega(ab)=\pi_A(a)\ox\pi_B(b)$ induces the irreducible representations $\pi_j:\A_j\to\B(\H_j)$, $j=A,B$, from \cref{it:normal_form} of \cref{thm:pure_tame_states}.
            In particular, tame bipartite pure states satisfy Haag-duality.
        \item\label{it:tame_GNS_marginal} 
            The GNS representation of the marginals can be computed from the tensor product splitting of $\pi_\omega$ as follows:
            Let $\Omega_\omega = \sum_{\alpha=1}^k \lambda_\alpha \Phi_\alpha^A\ox\Phi_\alpha^B$, $k\in\NN\cup\set\oo$, be the Schmidt decomposition of the GNS vector.
            Set $P_j=\sum_{\alpha=1}^k\kettbra{\Phi_\alpha^j}$. Then the GNS representation of $\omega_A$ is $\H_{\omega_A}=\H_A\ox P_B\H_B$, $\pi_{\omega_A}=\pi_A\ox P_B$ and $\Omega_{\omega_A}=\Omega_\omega$ and similarly for $\omega_B$.
    \end{enumerate} 
\end{cor}
\begin{proof}
    \ref{it:tame_GNS}:
    The tricky part here is to verify well-definedness (this is trivially true if $\A$ is a $C^*$-tensor product).
    Assume that $\A$ is the maximal tensor product, then $\pi_A\ox\pi_B$ is an irreducible representation that realizes $\omega$ as a vector state (induced by $\Omega$).
    Therefore the GNS representation is $(\H_A\ox\H_B,\pi_A\ox\pi_B,\Omega)$ if the bipartite algebra is $\A_{max}$.
    For a general bipartite algebra, we consider the canonical homomorphism $\phi:\A_{max}\to\A$ (see \cref{thm:bipartite_algebras}) and apply \cref{thm:bipartite_algebras}.\ref{it:GNS_pullback} which shows that $\pi_\omega\circ\phi = \pi_{\omega\circ \phi}=\pi_A\ox\pi_B$.
    This implies that $\pi_\omega(ab) = \pi_\omega(\phi(a\ox b))=\pi_A(a)\ox\pi_B(b)$ holds and, hence, is indeed well-defined as a mapping from $\A\to\B(\H_\omega)$.

    \ref{it:tame_GNS_marginal}: It is clear that $\pi_A(\placeholder)\ox P_B$ is a representation in which $\Omega_\omega$ realizes $\omega_A$ as a vector state.
    We only have to prove cyclicity.
    Note that $(\ketbra{\Phi}{\Phi_\beta^A}\ox P_B)\Omega = \lambda_\beta \Phi \ox \Phi^B_\beta$.
    Since $\pi_A$ is irreducible, we thus have $[\pi_A(\A_A)\ox P_B\Omega_\omega]=[\B(\H_A)\ox P_B \Psi] = \H_A\ox P_B\H_B$.
\end{proof}

We see that a pure bipartite state is tame if and only if it has the split property (this is false for mixed states in general).
For the proof of \cref{thm:pure_tame_states}, we need the following Lemma:

\begin{lem}\label{thm:pure_tame_states_lemma}
    Let $\omega$ be a pure state on a bipartite algebra. Then:
    \begin{enumerate}[(1)]
        \item\label{it:neither_or_both} $\M_A$ is type I if and only if $\M_B$ is type I.
        \item\label{it:type=type} The marginals $\omega_j:\A_j\to\CC$ are factor states.
            The type of the marginal $\omega_j$ is the same as the type of the factor $\M_j$.
    \end{enumerate} 
\end{lem}
\begin{proof}[Proof of the Lemma]
    \ref{it:neither_or_both}:
    Assume $\M_A$ is type I. Then there are Hilbert spaces $\V,\W$ so that $\H=\V\ox\W$ and $\M_A = \B(\V)\ox\1_\W$.
    Thus $\M_B= \1\ox \N$ for some factor $\N\subset\B(\W)$ the irreducibility implies $\N=\B(\W)$ and hence $\M_B=\M_A'$.

    \ref{it:type=type}:
    The embedding $\A_j\to\A$ induces an isometry between the GNS spaces $V_j:\H_{\omega_j}\to\H_\omega$. Consider the projection $P_j = V_jV_j^*\in \M_j'$ which projects onto $[\M_j\Omega]\cong \H_{\omega_j}$. 
    Then $\pi_{\omega_j}(\A_j)''\cong P_j \M_j P_j$ which has the same type as $\M_j$ (this can be seen from combining \cite[Cor.~5.5.7]{kadisonringrose1} and \cite[Ex.~6.9.16]{kadisonringrose2}).
\end{proof}

\begin{proof}[Proof of \cref{thm:pure_tame_states}]
    We denote the bipartite algebra on which $\omega$ acts by $\A$.
    The implication \ref{it:normal_form} $\Rightarrow$ \ref{it:tame} is trivial, \ref{it:normal_form} $\Rightarrow$ \ref{it:both_typeI} follows from \cref{it:type=type} of \cref{thm:pure_tame_states_lemma} and \ref{it:both_typeI} $\Leftrightarrow$ \ref{it:one_typeI} is seen from combining \cref{it:neither_or_both,it:type=type} of \cref{thm:pure_tame_states_lemma}.

    \ref{it:tame} $\Rightarrow$ \ref{it:normal_form}: 
    By \cref{thm:mixed_tame} we may assume $\A=\A_{min}$ without loss of generality.
    Denote by $Q$ the set of bipartite states $\omega$ on $\A_{min}$ which can be written as $\omega(a\ox b)=\tr[\rho(\tilde\pi_A(x)\ox\tilde\pi_B(x))]$ for representations $\tilde\pi_j$ and a density operator $\rho$.
    Let $(X,\mu)$ and $x\mapsto\omega_x$ be as in \cref{def:tame} (note that $\omega_x\in Q$).
    Denote by $\nu$ the push-forward measure of $\mu$ with respect to $x\mapsto \omega_x$. Then $\nu$ is a Radon probability measure on $\states(\A_{min})$ whose barycenter is the state $\omega$.
    Since $\omega$ is pure, we have $\nu=\delta_\omega$. This can only be true if $\omega_x=\omega$ holds $\mu$-almost everywhere.
    The only relevant part for us is that this implies $\omega\in Q$ so that there are representations $\tilde\pi_j:\A_j\to\tilde\H_j$ and a density operator $\rho$ on $\tilde\H_A\ox\tilde\H_B$ such that $\omega=\tr[\rho(\tilde\pi_{A}\ox\tilde\pi_{B}(\placeholder))]$.
    Because of the purity of $\omega$, we may replace $\rho$ by a vector state $\Omega\in\tilde\H_A\ox\tilde\H_B$.
    Let $\Omega=\sum_{\alpha=1}^k \lambda_\alpha \Phi_\alpha^A \ox \Phi_\alpha^B$ be its Schmidt decomposition and consider the invariant subspaces $\H_j = [\pi_j(\A_j)\{\Phi_1^j,\ldots,\Phi_k^j\}]$.
    We denote by $\pi_j$ the restriction of $\tilde\pi_j$ to $\H_j$.
    Then $\Omega$ is a cyclic vector for $\pi_A\ox\pi_B$ which implies that $\pi_A\ox\pi_B$ is the GNS representation of $\omega$.
    Therefore $\pi_A\ox\pi_B$ and, hence, $\pi_A$ and $\pi_B$ are irreducible representations.

    \ref{it:both_typeI} $\Rightarrow$ \ref{it:normal_form}: 
    By \cref{thm:pure_tame_states_lemma} both $\M_A$ and $\M_B$ are type I factors. 
    It follows that $\H = \H_A\ox\H_B$ with $\M_A=\B(\H_A)\ox\1$ and $\M_B= \1\ox\B(\H_B)$.
    Therefore we get representations $\pi_j:\A_j\to\B(\H_j)$ so that $\pi(ab) = \pi_A(a)\ox\pi_B(b)$. 
    These representations have to be irreducible as otherwise $\pi$ could not be irreducible.

    \ref{it:quasi_equiv} $\Rightarrow$ \ref{it:both_typeI}: Since $\omega$ is quasi-equivalent to a product state $\rho$, we have a normal $^*$-isomorphism $\phi:\pi_{\omega}(\A)''\to\pi_{\rho}(\A)''$ such that $\phi\circ\pi_{\omega}=\pi_{\rho}$.  \cref{lem:productstate} implies $\pi_{\rho}(\A)''\cong\pi_{\rho_{A}}(\A_{A})''\!\ \Bar\ox\!\ \pi_{\rho_{B}}(\A_{B})''$, such that $\pi_{\rho_{j}}(\A_{j})''\!\ \Bar\ox\!\ \CC\cong_{\phi}\M_{j}$, $j=A,B$. 
    The argument in the proof of \cref{it:type=type} in \cref{thm:pure_tame_states_lemma} shows that $\pi_{\omega_{j}}(\A_{j})''$ and $\M_{j}$ have the same type for $j=A,B$. Finally, by purity of $\omega$ we know that $\pi_{\omega}(\A)''\cong\B(\H_{\omega})$, which implies that $\M_{j}$, and, thus, $\pi_{\omega_{j}}(\A_{j})''$ is of type I for $j=A,B$.
    In particular, $\pi_\omega(\A)''=\pi_\varphi(\A)''$ so that $\varphi$ is quasi-equivalent to $\omega$.

    \ref{it:normal_form} $\Rightarrow$ \ref{it:quasi_equiv}:
    The proof of \cref{thm:tame_GNS} only uses \ref{it:normal_form}, so that we may use the result that the GNS representation of $\omega$ is given by $\pi_\omega(ab)=\pi_{A}(a)\ox\pi_B(b)$ on $\H_\omega=\H_A\ox\H_B$.
    Consider arbitrary unit vectors $\Phi_j\in\H_j$ and set $\varphi_j = \ip{\Phi_j}{\pi_j(\placeholder)\Phi_j}$.
    We get a product state $\varphi$ on $\A$ by $\varphi(ab)=\varphi_A(a)\varphi_B(b)=\ip{\Phi_A\ox\Phi_B}{\pi_\omega(ab)(\Phi_A\ox\Phi_B)}$ (the second equality shows well-definedness).
    Invoking \cref{thm:bipartite_algebras}.\ref{it:GNS_pullback} the GNS representation of $\varphi$ may be constructed from $\H_\varphi=\H_{\varphi_A}\ox\H_{\varphi_B}$, $\pi_\varphi(ab)=\pi_A\ox\pi_B$, and $\Omega_{\varphi_A}\ox\Omega_{\varphi_B}$.
    Clearly, $\pi_\varphi(\A)=\pi_\omega(\A)$ which, in particular, implies quasi-equivalence with $\omega$.
\end{proof}

It turns out that the class of unital $C^*$-algebras with the property that all pure bipartite states with arbitrary other systems are tame coincides with a well-known class of $C^*$-algebras known as type I $C^*$-algebras (see \cite{glimm1961type}, \cite{Arveson1976} or \cite[Ch.~9]{dixmier1982}).
The defining property of a type I $C^*$-algebra $\A$ is that for every representation $\pi:\A\to\B(\H)$ the generated von Neumann algebra $\pi(\A)''\subset\B(\H)$ has type I.

\begin{cor}\label{thm:typeI_cstar}
    Let $\A_A$ be a $C^*$-algebra. The following are equivalent
    \begin{enumerate}[(a)]
        \item\label{it:all_bipartite_tame} For every algebra $\A_B$, every pure bipartite state is tame.
        \item\label{it:typeI_algebra} $\A_A$ is a type I $C^*$-algebra.
    \end{enumerate}
\end{cor}

\begin{proof}
    \ref{it:typeI_algebra} $\Rightarrow$ \ref{it:all_bipartite_tame}: If $\A_A$ is type I and $\omega$ is a bipartite state, then $\omega_A$ is always a type I state, hence $\omega$ is tame.
    \ref{it:all_bipartite_tame} $\Rightarrow$ \ref{it:typeI_algebra}:
    It suffices to show that every factor representation $\pi$ generates a type I factor \cite[Add.~9.5.9]{dixmier1982}.
    We start by showing for every factor representation $\pi$ on a Hilbert space $\H$ there a $C^*$-algebra $\A_B$ and a bipartite pure state $\omega$ such that $\H_\omega=\H$ and $\pi_\omega|_{\A_A} = \pi$.
    We set $\A_B=\pi(\A)'$ and define $\omega(a\ox b) = \ip{\Omega}{\pi(a)b\Omega}$ for any unit vector $\Omega\in\H$.
    Then $\pi(\A)''\vee \A_B = \B(\H)$ because $\pi$ is a factor representation so that $\rho(a\ox b) = \pi(a)b$ extends to an irreducible representation of $\A_A\ox_{max}\A_B$.
    This implies that $(\pi_\omega,\H_\omega,\Omega_\omega)=(\rho,\H,\Omega)$ so that $\omega$ is pure because its GNS representation is irreducible.
    By assumption, $\omega$ is tame, and thus the marginal $\omega_A$ is a type I state.
    By \cref{thm:pure_tame_states_lemma} this implies that $\pi_\omega(\A_A)''=\pi(\A)''$ is a type I factor which proves that $\pi$ is a type I factor representation.
\end{proof}

\begin{rem}
    Apart from "tame" and "wild" being standard terms to emphasize a certain kind of regularity (or the lack thereof), they appear in von Neumann algebra theory.
    A factor is called tame if it is type I and otherwise wild.
    For pure bipartite states, we saw in \cref{thm:pure_tame_states} that there is a direct connection between tame (resp.\ wild) factors and tame (resp.\ wild) bipartite states.
    Even for mixed states, there is a connection: A bipartite state is mixed if it can be obtained through shared randomness and tame bipartite states between Alice and Bob.
\end{rem}

\begin{rem}
    The obtained results can also be used to give a new proof for the well-known fact that type I $C^*$-algebras are nuclear. 
    If $\A_A$ is a type I $C^*$-algebra and $\A_B$ is any unital $C^*$-algebra, then every pure state $\omega$ on $\A_A\ox_{max}\A_A$ is necessarily tame (because $\omega_A$ is a type I state). 
    Therefore all pure states on the maximal $C^*$-tensor product factorize through the minimal $C^*$-tensor product.
    By the Krein-Milman theorem this implies that all states of $\A\ox_{max}\B$ factor through the minimal $C^*$-tensor product which implies that $\A\ox_{max}\B=\A\ox_{min}\B$.
\end{rem}

\subsection{Connection to Tsirelson's problem}\label{sec:Tsirelsons_prob}

We close this section by making explicit the connection between our approach to the commuting operator framework and the work on correlation bodies in the context of Tsirelson's problem.
The setting is that Alice and Bob are free to choose $n$ different measurements with $k$ outcomes each where $n$ and $k$ are finite integers.
The correlation functions
\begin{equation*}
    p(\alpha,\beta|i,j)
\end{equation*}
describe the probability of Alice measuring outcome $\alpha$ and Bob measuring $\beta$ given that they respectively pick measurements $i$ and $j$.
There are different ways to model this setup mathematically, giving rise to different sets of correlation functions.
These sets are convex subsets $\C_*\subset \RR^{k^2n^2}$ with an index representing the framework in which composite systems are modeled.

The simplest model is finite-dimensional quantum mechanics where Alice's POVMs $\set{M_{i,\alpha}^A}$ act on a finite-dimensional Hilbert space $\H_A$ and Bobs POVMs $\set{M_{j,\beta}^B}$ act on a finite-dimensional Hilbert space $\H_B$.
The state of the joint system would be described by a density operator $\rho$ on the joint Hilbert space $\H_{AB}=\H_A\ox\H_B$ yielding $p(\alpha,\beta|i,j) = \tr[\rho\, M_{i,\alpha}^A\ox M_{j,\beta}^B]$.
The set of correlation functions which can be obtained with finite-dimensional quantum mechanics is denoted $\C_q$.

Slofstra showed in \cite{Slofstra} (see also \cite{Dykema2019}) that $\C_q$ is, in general, not closed.
Its closure is another correlation body denoted by $\C_{qa}$ ("a" is short for "approximation").
If we allow for infinite-dimensional Hilbert spaces $\H_A$ and $\H_B$ and model the bipartite state by a density operator $\rho$, we obtain a correlation body denoted $\C_{qs}$ (the "s" is short for "spatial").
The correlations that one obtains by generalizing from density operators on $\H_A\ox\H_B$ to mere algebraic states on $\rho:\B(\H_A\ox\H_B)\to\CC$, i.e.\ $p(\alpha,\beta|i,j)=\rho(M_{i,\alpha}^A\ox M_{j,\beta}^B)$, are precisely those that finite-dimensional ones can approximate (this follows from \cref{thm:mixed_tame}).
Therefore, $\C_{qa}$ consists of infinite-dimensional correlations where one requires a Hilbert space tensor product $\H_A\ox\H_B$ but allows for singular states.

The most general correlation body is obtained by dropping the requirement of a tensor product separation. Instead, one only requires that $M_{i,\alpha}^A$ and $M_{j,\beta}^B$ are commuting POVMs on a joint Hilbert space $\H_{AB}$. 
The resulting correlation functions $p(\alpha,\beta|i,j)=\tr[\rho M_{i,\alpha}^AM_{j,\beta}^B]$, where $\rho$ is a density operator on $\H_{AB}$, are called the commuting operator correlations. We collect them in a set $\C_{qc}$.
One does not gain more correlations by admitting singular states with respect to commuting observables, as the GNS representation shows.

Since these different frameworks of modeling correlation experiments are increasing in generality, it holds that 
\begin{equation}
    \C_q\subset \C_{qs}\subset \C_{qa} \subset \C_{qc}. 
\end{equation}
By now, each of these inclusions is known to be strict (for sufficiently large $n$ and $k$).
It was proved in \cite{Slofstra,Dykema2019} that $\C_{qs} \subsetneq\C_{qa}$.
To decide the strictness of the inclusion $\C_{qa}\subset\C_{qc}$, i.e.\ to decide whether all commuting operator correlations can be approximated by finite-dimensional ones, is known as Tsirelson's problem \cite{tsirelsons_prob,scholz2008tsirelson} and was famously solved in \cite{Ji2020}.

Consider the universal $C^*$-algebras $\A_A$ and $\A_B$ generated by $n\cdot k$ positive operators $m_{i,\alpha}^j$, $i=1,\ldots,n$, $\alpha=1,\ldots,k$ with the relations $\sum_\alpha m_{i,\alpha}^j = 1$ for all $i$ where $j=A,B$.
The different notions of bipartite states that we discussed are directly connected to the different correlation bodies:
\begin{thm}\label{thm:tsirelsons_prob}
    Let $\omega$ be a pure state on $\A_A\ox_{max}\A_B$ and set $p(\alpha,\beta|i,j)=\omega(m_{i,\alpha}^Am_{j,\beta}^B)$.
    Then $p$ is a commuting operator framework correlation function, i.e.\  $p\in\C_{qc}$. Furthermore,
    \begin{enumerate}[(1)]
        \item\label{it:Cq} $p\in\C_q$ if and only if $\omega$ has finite Schmidt rank,
        \item\label{it:Cqs} $p\in\C_{qs}$ if and only if $\omega$ is tame,
        \item\label{it:Cqa} $p\in\C_{qa}$ if and only if $\omega$ factors through the minimal $C^*$-tensor product $\A_A\ox_{min}\A_B$.
    \end{enumerate} 
\end{thm}

The first item makes reference to the Schmidt rank, which will be introduced formally in the next section.
These equivalences are, however, irrelevant here as we only need the algebraic definition in the proof, which assigns to $\omega$ the Schmidt rank of the GNS vector provided that $\omega$ is tame.

\begin{proof}
    That $p\in\C_{qc}$ follows by applying the GNS representation.
    The "only if" parts of items \ref{it:Cq}, \ref{it:Cqs} and \ref{it:Cqa} rely on the basic fact that every collection of POVMs $M_{i,\alpha}^A$ and $M_{j,\beta}^B$ on Hilbert spaces induces representations $\pi_j$ on $\A_j$ sending $m_{i,\alpha}^A$ to $M_{i,\alpha}^A$ and similarly for $B$.

    \ref{it:Cq} follows from the algebraic definition of the Schmidt rank (i.e.\ item \ref{it:alg} of \cref{thm:schmidt_rank}).

    \ref{it:Cqs} follows from \cref{thm:pure_tame_states}.

    \ref{it:Cqa}: 
    This item is true also if $\omega$ is not pure. Recall that $\C_{qa}=\Bar{\C_q}$. 
    Denote by $\C_{qmin}$ the correlation functions induced by states on the minimal $C^*$-tensor product.
    "only if": Denote by $\S\subset \A_A\ox_{min}\A_B$ the self-adjoint unital subspace spanned by elements of the form $m_{i,\alpha}^Am_{j,\beta}^B$.
    Every $p\in\C_{qa}$ defines a state on $\S$, and by Arveson's extension theorem \cite{paulsen2002completely}, we can extend this to a state on the minimal tensor product, showing $\C_{qa}\subseteq \C_{qmin}$.
    "if": Let $p\in\C_{qmin}$.
    By \cref{thm:min_states} we can find representations $\pi_j:\A_j\to\B(\H_j)$ and a possibly singular state $\eta:\B(\H_A\ox\H_B)\to\CC$ so that $\eta(\pi_A(m_{i,\alpha}^A)\ox\pi_B(m_{j,\beta}^B))=p(\alpha,\beta|i,j)$.
    We may $w^*$-approximate the singular state $\eta$ by density operators with finite rank.
    Then we can also restrict POVMs to the ranges of these density operators to get correlation functions in $\C_q$, which approximate $p$.
\end{proof}

Let us comment on the statements in the theorem in cases where $\omega$ is mixed.
The first item only makes sense for pure states as we did not give a definition of the Schmidt rank for mixed states.
The third item also holds if $\omega$ is not pure, as the above proof shows.
However, the equivalence in the second item breaks down for general mixed states.
Instead, it only holds that $p\in \C_{qs}$ implies that $\omega$ is tame.
The reason is that in $\C_{qs}$, the shared randomness between Alice and Bob that is needed to write a mixed state as a mixture of pure ones must be cast into a direct sum of Hilbert spaces.
This is, however, not always possible. 
Our definition of tame states is more liberal as it allows for a classical system where all shared randomness may be stored.
It guarantees that all separable states are tame, whereas it seems to us that $p\in\C_{qs}$ is not guaranteed if $\omega$ is a bipartite separable state on $\A_A\ox_{min}\A_B$.

In view of a recent contribution \cite{cabello2023} to the physical implications of $\C_{qa}\ne\C_{qc}$, we mention that it is impossible to obtain correlation functions in $\C_{qc}\setminus\C_{qa}$ in quantum field theories with hyperfinite local observable algebras if the two systems are spacelike-separated.
Essentially, this was already observed in \cite{scholz2008tsirelson}, but no proof was given.
We include the formal statement and a proof in \cref{sec:qft}, where bipartite states arising from quantum field theories are discussed in more detail.

\section{The Schmidt rank for pure bipartite states}\label{sec:Schmidt}

\subsection{Equivalent definitions}

We generalize the Schmidt rank to bipartite states on general bipartite algebras as introduced in \cref{sec:cof}.
This section is self-contained (we repeat all definitions and results stated in \cref{sec:summary} regarding the Schmidt rank), we focus on the mathematical theory and refer to the discussion in \cref{sec:summary} for more physical intuitions.

Let us start with the concept of compressions relative to a bipartite state.
As explained in \cref{sec:summary}, these are compressions of observables corresponding to state emulations in the Schr\"odinger picture.

\begin{defin}
    Let $\omega$ be a pure  bipartite state.
    A {\bf compression} with respect $\omega$ is a collection $(C_A,C_B,\K,\Psi)$ where $\K$ is a Hilbert space $C_A$ and $C_B$ are unital completely positive maps $C_j : \A_j \to \B(\K)$, $j=A,B$, with commuting ranges and $\Psi\in\K$ is a unit vector such that 
    \begin{equation}
        \omega(ab) = \ip\Psi{C_A(a)C_B(b)\Psi}\qquad\forall(a,b)\in\A_A\times\A_B.
    \end{equation}
\end{defin}

If $(C_A,C_B,\K,\Psi)$ is a compression, then there is a bipartite system acting on $\K$ with Alice and Bob's algebras $\B_A$ and $\B_B$ being the $C^*$-algebras generated by $C(\A_A)$ and $C(\A_B)$ and with the bipartite algebra $\B$ being the $C^*$-algebra generated by $\B_A$ and $\B_B$.
The pure bipartite state $\psi=\ip\Psi{(\placeholder)\Psi}:\B\to\CC$ on this system is able to emulate all correlations of $\omega$.
By the universal property of the maximal $C^*$-tensor product (see \cref{thm:univ_prop}), the equation
\begin{equation}\label{eq:compression_C}
    C(a\ox b) =C_A(a)C_B(b)
\end{equation}
 determines a bijection between tuples $(C_A,C_B)$ of unital completely positive maps $C_j:\A_j\to\B(\K)$, $j=A,B$, with commuting ranges and unital completely positive maps $C:\A_{max}\to \B(\K)$ with the property that $C(\A_A\ox1)$ commutes with $C(1\ox\A_B)$.
In the following, if $(C_A,C_B,\K,\Psi)$ is a compression, we denote by $C$ this map.
Note that this makes sense, regardless of the bipartite algebra on which the bipartite state is defined. The GNS construction shows that compressions always exist, i.e., the reduced GNS representation $(\pi_{\omega|\A_{A}}, \pi_{\omega|\A_{B}}, \H_{\omega}, \Omega_{\omega})$ is a compression with respect to $\omega$ (see \cref{thm:correlation_invariants}).

For a correlation experiment, it is not necessary that both parties measure simultaneously.
Let us pretend that Alice starts by measuring an effect $0\le a\le1$. Then we should describe Bob's measurement results by the subnormalized state $\omega(a(\placeholder))$.
It turns out that the Schmidt rank is closely related to the factorization properties of the completely positive map  
\begin{equation}\label{eq:state_cp}
     \Gamma_\omega:\A_A\to\A_B^*,\ a\mapsto\omega(a(\placeholder)).
\end{equation}
If Bob were to measure first, we would instead get a map $\A_B\to\A_A^*$, which is nothing else than (the restriction to $\A_B$ of) the transpose of $\Gamma_\omega$.
Conversely, every completely positive map $\Gamma:\A_A\to\A_B^*$, with the property that $1\in\A_A$ is mapped to a state, uniquely determines a bipartite state $\omega_\Gamma(ab)=\Gamma(a)(b)$ up to correlation-equivalence.

The last ingredient needed for stating the main theorem of this section is the {\bf rank of a state} on a $C^*$-algebra which will apply to the marginals.
This is motivated by the observation that in finite dimensions, the Schmidt rank of a pure bipartite state coincides with the rank of the reduced density operators.
We define the rank of a state $\varphi$ on a $C^*$-algebra as the smallest number $n$ so that $\varphi$ can be written as a convex combination of $n$ pure states, and we define the rank to be infinite if no such $n$ exists, i.e.\
\begin{equation}\label{eq:rank}
    \rank(\varphi) = \min\set[\big]{n,\oo\given \text{$\varphi$ is a convex combination of $n$ pure states}}
\end{equation}
Clearly, this definition agrees with the rank of a density operator in finite-dimensional systems.

With these concepts, we can now state the main theorem:

\begin{thm}\label{thm:schmidt_rank}
    Let $\omega$ be a pure bipartite state. The following definitions for the Schmidt rank $\SR(\omega)\in\NN\cup\{\oo\}$ are equivalent:
    \begin{enumerate}[(A)]
        \item\label[defin]{it:alg} 
            If $\omega$ is tame, $\SR(\omega)$ is the Schmidt rank of the GNS vector $\Omega_\omega$ with respect to the tensor splitting of the GNS space $\H_\omega$ (see \cref{thm:tame_GNS}).
            If $\omega$ is wild, $\SR(\omega)=\oo$,
        \item\label[defin]{it:op_cof}
            The Schmidt rank is the square root of the minimal attainable Hilbert space dimension for a compression $(C_A,C_B,\K,\Psi)$ with respect to $\omega$, i.e., $\SR(\omega):=\min \sqrt{\dim(\K)}$,
        \item\label[defin]{it:op_split}
            The Schmidt rank of $\omega$ is the smallest number $k\in\NN$ which admits unital completely positive maps $C_j:\A_j\to\MM_k$ and a unit vector $\Psi\in\CC^k\ox\CC^k$ such that $\omega(ab)=\ip\Psi{C_A(a)\ox C_B(b)\Psi}$ for all $(a,b)\in\A_A\times\A_B$. 
            If no such number exists, the Schmidt rank is infinite,
        \item\label[defin]{it:factor}
            The Schmidt rank of $\omega$ is the smallest number $k\in\NN$ so that the completely positive map $\Gamma_\omega:\A_A\to\A_B^*$ factorizes through $\MM_k$ in the sense that there are completely positive maps $\alpha:\A_A\to\MM_k$ and $\beta:\MM_k\to\A_B^*$ so that $\Gamma_\omega=\beta\circ\alpha$, i.e.\ the following diagram commutes:
            \begin{equation}
                \begin{tikzcd}
                    \A_A \arrow{rr}{\Gamma_\omega} \arrow{dr}{\alpha} &&\A_B^*\\& \MM_k\arrow{ur}{\beta} &
                \end{tikzcd}
            \end{equation}
            If no such number exists, the Schmidt rank is infinite.
        \item\label[defin]{it:rank}
            The Schmidt rank is the rank of the reduced states $\SR(\omega):=\rank(\omega_A)=\rank(\omega_B)$.
        \item\label[defin]{it:interval} 
            The Schmidt rank of  $\omega$ is the square root of the dimension of the order interval of the marginals, i.e.\ $\SR(\omega):=\sqrt{\dim\,[0,\omega_A]}=\sqrt{\dim\,[0,\omega_B]}$.  By the dimension of the order interval, we mean the vector space dimension of its real linear hull.
    \end{enumerate} 
\end{thm}

\cref{it:alg} agrees with the usual definition of the Schmidt rank if $\A_j = \B(\H_j)$ and $\omega(a\ox b)= \ip\Psi{a\ox b\Psi}$ with $\Psi\in\H_A\ox\H_B$. This is because the GNS representation is just the identity, so the Schmidt rank of $\omega$ is the vector Schmidt rank of $\Omega_\omega=\Psi$.
In particular, all states of finite Schmidt rank are tame.
Combining this with the results from \cref{sec:cof}, we find the following hierarchy of increasing entanglement for pure states:
\begin{equation}\label{eq:hierachy}
    \Big\{\!\!\begin{array}{c}\text{\small product}\\[-.1cm] \text{\small states}\end{array}\!\!\Big\}=\SR_1 \subset \SR_{\le 2}\subset \ldots \subset \states_{tame}\subset\states_{min}\subset\states_{max}.
\end{equation}
\cref{it:op_split,it:op_cof} are similar. In \cref{it:op_split}, we only look at finite-dimensional compressions with a tensor splitting while \cref{it:op_cof} allows for general commuting operator compressions.
Essentially their equivalence follows because the commuting operator framework is equivalent to the tensor product formalism in finite-dimensions.
Differences only occur in cases where the Schmidt rank is infinite and are not seen by the definitions.

Finally, we discuss \cref{it:factor}.
While we only ask that the factorizing maps $\alpha$ and $\beta$ are completely positive in the definition, we discussed a physical interpretation of this definition in terms of encoding and decoding operations in \cref{sec:summary}, which requires $\beta$ to be a quantum channel.
This is justified by the following:

\begin{lem}\label{thm:factorization}
    The following are equivalent: 
    There is a factorization $\Gamma_\omega= \beta\circ\alpha$ for maps $\alpha:\A_A\to\MM_k$ and $\beta:\MM_k\to\A_B^*$ such that
    \begin{enumerate}[(a)]
        \item\label{it:cp_factor} $\alpha$ and $\beta$ are completely positive.
        \item\label{it:cp_factorB} $\alpha$ and $\beta$ are completely positive, $\alpha$ is unital and $\beta(1)$ is a state.
        \item\label{it:cp_factorC} $\alpha$ and $\beta$ are completely positive, $\alpha(1)$ is a density operator and $\beta$ takes density operators to states.
    \end{enumerate} 
\end{lem}
\begin{proof}
    \ref{it:cp_factor} $\Leftrightarrow$ \ref{it:cp_factorB} follows from the fact that $\alpha = K\tilde\alpha(\placeholder)K$ for a unital completely positive map $\tilde \alpha$ and $K=\alpha(1)^{\frac{1}{2}}$ \cite[Lem.~2.2.5]{brownozawa} if we replace $\alpha$ and $\beta$ by $\tilde\alpha$ and $\tilde\beta=\beta(K(\placeholder)K)$.

    \ref{it:cp_factor} $\Leftrightarrow$ \ref{it:cp_factorC} can be deduced using the same trick but applied to the transpose maps $\Gamma_\omega^*=\alpha^*\circ\beta^*$.
    Identifying $\MM_k$ with its dual, we may thus assume that $\beta^*$ is unital and that $\alpha^*(1)$ is a state which implies the desired properties for $\alpha$ and $\beta$. 
\end{proof}

\subsection{The case where one system is finite-dimensional}

If one of the two parties has a finite-dimensional quantum system, we can immediately compute the Schmidt rank. 
This situation is important in applications to device-independent cryptography (see \cref{sec:summary}), and the results that we obtain turn out to be a helpful tool in the proof of \cref{thm:schmidt_rank}.
This situation allows for a normal form involving a one-sided compression of the large system. This turns out to be a helpful tool in the proof of \cref{thm:schmidt_rank}.

Let $\A_A=\MM_n$ and let $\A_B$ be an arbitrary $C^*$-algebra.
Then there is a unique bipartite algebra $\A=\MM_n\ox\A_B \equiv \MM_n(\A_B)$ so that we are guaranteed that $A$ and $B$ are statistically independent (see \cref{thm:bipartite_algebras}).

\begin{prop}\label{thm:one_sys_finite}
    Let $\omega$ be a bipartite pure state for $\A_A=\MM_n$ and arbitrary $\A_B$.
    Denote $\rho$, the density operator that implements the reduced state $\omega_A$ on $\MM_n$.
    Set $k=\rank(\rho)$ and let $\Psi \in\CC^n\ox\CC^k$ be the canonical purification of $\rho$.
    Then there is a unique unital completely positive map $T_B:\A_B\to\MM_k$ such that 
    \begin{equation}
        \ip\Psi{a\ox T_B(b)\Psi} =\omega(a\ox b)\qquad\forall(a,b)\in\MM_n\times\A_B.
    \end{equation}
\end{prop}

With \cref{it:rank}, it is clear that the Schmidt rank of $\omega$ is just the rank of the density operator $\rho$.
The unital completely positive map $T_B$ can be constructed explicitly: Let $\rho = \sum_{i=1}^k p_i \kettbra{\Phi_i} $, and pick the purification $\Psi = \sum_{i=1}^k \sqrt{p_i}\,\Phi_i\ox\ket i$. Then $\bra iT_B(b)\ket j = (p_ip_j)^{-1/2} \omega(\ketbra{\Phi_i}{\Phi_j}  \ox b)$.
\cref{thm:one_sys_finite} follows directly from:

\begin{lem}["Bob joins Alice's GNS space"]\label{thm:bob_joins_alice_GNS}
    Let $\omega$ be a (not necessarily pure) state on a bipartite algebra. Pick one marginal, say $\omega_A$.
    There is a unique operator $T_B(b) \in \pi_{\omega_A}(\A_A)'$, such that 
    \begin{equation}\label{eq:bob_joins_alice_GNS}
        \omega(ab) = \ip{\Omega_{\omega_A}}{\pi_{\omega_A}(a)T_B(b)\Omega_{\omega_A}}\qquad\forall(a,b)\in\MM_n\times\A_B.
    \end{equation}
    The map $T_B:\A_B\to\pi_{\omega_A}(\A_A)'$ is unital and completely positive.
    A dilation of $T_B$ is given by the isometry $V:\H_{\omega_A}\to \H_\omega$ induced by the embedding $\A_A\hookrightarrow\A$%
    \footnote{Let $\phi:\B\to\C$ be a $^*$-homomorphism and let $\omega$ be a state on  $\C$. Then the induced isometry $V:\H_{\omega\circ\phi}\to\H_\omega$ is defined by $V\pi_{\omega\circ\phi}(b)\Omega_{\omega\circ\phi}=\pi_\omega(\phi(b))\Omega_\omega$.}
    and the representation $\pi_\omega|_{\A_B}$:
    \begin{equation}\label{eq:TB_dilation}
        T_B = V^* \pi_\omega(\placeholder)V.
    \end{equation}
\end{lem}

\begin{proof}
    It is clear that \cref{eq:TB_dilation} defines a unital completely positive map from $\A_B$ to $\B(\H_{\omega_A})$.
    Note that $\pi_{\omega_A}(a)=V^*\pi(a)V$.
    We check that \cref{eq:bob_joins_alice_GNS} holds:
    \begin{align*}
        \omega(ab) = \ip{\Omega_\omega}{\pi_\omega(a)\pi_\omega(b)\Omega} 
        &= \ip{\Omega}{V^*V\pi_\omega(a)VV^*\pi(b) V^*V\Omega} \\
        &= \ip{\Omega_{\omega_A}}{\pi_{\omega_A}(a)T_B(b)\Omega_{\omega_A}}.
    \end{align*}
    That $T_B(b)\in\pi_{\omega_A}(\A_A)'$ holds follows from the formula $\pi_\omega(a)T_B(b)=V^*\pi(ab)V$ which holds because $T_B(b)\pi_{\omega_A}(a) = V^* \pi_\omega(b)VV^*\pi_\omega(a)V = V^*\pi_\omega(b)\pi_\omega(a)V= V^*\pi(ab)V$.
    Finally, uniqueness holds because every operator $T_B(b)\in\pi_{\omega_A}(\A_A)'$ which satisfies \eqref{eq:bob_joins_alice_GNS}, has the same matrix elements with respect to the dense subspace $\pi_{\omega_A}(\A_A)\Omega_{\omega_A}$:
    \begin{align*}
        \ip{\pi_{\omega_A}(a_1)\Omega_{\omega_A}}{T_B(b)\pi_{\omega_A}(a_2)\Omega_{\omega_A}}
        =\ip{\Omega_{\omega_A}}{\pi_{\omega_A}(a_1^*a_2)T_B(b)\Omega_{\omega_A}}
        =\omega(a_1^*a_2b).
    \end{align*}\qedhere
\end{proof}

For tame bipartite pure states, we get an explicit construction of the map $T_B$ appearing in \cref{thm:bob_joins_alice_GNS}:
Let $\omega$ be a tame pure bipartite state with Schmidt rank $k\in\NN\cup\{\oo\}$. 
Consider the Schmidt decomposition $\Omega_\omega = \sum_{j=1}^k \lambda_\alpha \,\Phi_\alpha^A\ox \Phi_\alpha^B \in \H_\omega$ with respect to the canonical tensor product decomposition $\H_\omega=\H_A\ox\H_B$ of its GNS space.
Let $\pi_j:\A_j\to\B(\H_j)$ be the irreducible representations such that $\pi_\omega(ab)=\pi_A(a)\ox\pi_B(b)$.
By \cref{thm:tame_GNS}, the GNS representation of the reduced state $\omega_A$ is
\begin{equation}
    \H_{\omega_A}=\H_A\ox P_B\H_B,\quad \pi_{\omega_A} = \pi_A(\placeholder)\ox P_B, \quad\Omega_{\omega_A}=\Omega_\omega,
\end{equation}
where $P_B = \sum_{j=1}^k \kettbra{\Phi_\alpha^B}$.
The unique unital completely positive map $T_B$ is 
\begin{equation}
    T_B = \1_{\H_A}\ox (P_B\pi_B(\placeholder)P_B).
\end{equation}
We see that the range of $T_B$ is weakly dense in $\M_B=\1\ox\B(\H_B)$.

\subsection{Proof of equivalence}\label{sec:proof}
\newcommand\SRdef[1]{\SR^{\textup{\ref*{#1}}}}

In this subsection, we prove \cref{thm:schmidt_rank}.
For the sake of this proof, we introduce the notation $\SR^{\textup{(X)}}(\omega)$ for the Schmidt rank in the sense of Definition (X) where (X) $=$ \ref*{it:alg},\,\ldots,\,\ref*{it:interval}.

A tool used repeatedly in the proof is the Radon-Nikodym theorem for completely positive maps (in particular, for states).
We briefly explain the theorem here but refer to \cref{sec:appendix} for a more detailed explanation, including a full proof.
If $\B$ is a $C^*$-algebra and $S,T:\B\to\B(\H)$ are completely positive maps, then we write $S\le_{cp}T$ if $T-S$ is completely positive and we denote by $[0,T]_{cp}$ the convex set of completely positive maps $S\le_{cp}T$.
Let $T = V^*\pi(\placeholder)V$ be the minimal Stinespring dilation.
The Radon-Nikodym theorem asserts 
\begin{equation}\label{eq:radon_niko_cp}
    S\leftrightarrow Q \iff S = V^*\pi(\placeholder) Q V
\end{equation}
defines a monotone and affine bijection between $[0,T]_{cp}$ and $[0,\1]_{\pi(\B)'}$.
The most important special case is that of a state $\omega:\B\to\CC$. 
Since the Stinespring dilation of a state is its GNS representation, the bijection between $[0,\omega]$ and $[0,\1]_{\pi_\omega(\B)'}$ is
 \begin{equation}\label{eq:radon_niko}
     \varphi\leftrightarrow Q \iff \varphi = \ip{\Omega_\omega}{\pi_\omega(\placeholder)Q\Omega_\omega}.
\end{equation}
In the proof of \cref{thm:schmidt_rank}, we will use that the linear extension of this bijection is a complete order embedding of $\pi_\omega(\B)'$ into $\A^*$.

\subsubsection*{Step 1. Equivalence of Definitions \ref{it:alg}, \ref{it:op_cof} and \ref{it:op_split}} 

Since Definitions \ref{it:alg}, \ref{it:op_cof} and \ref{it:op_split} make no reference to the bipartite algebra, we may, without loss of generality, assume it to be the maximal $C^*$-tensor product $\A_{max}$.

We start with a construction that takes in a finite-dimensional compression and returns another one with a smaller dimension and with a tensor-splitting structure.
Its properties are summarized in the following:

\begin{prop}\label{thm:finite_compression}
    Let $\omega$ be a pure bipartite state. 
    If a finite-dimensional compression $(C_A',C_B',\K',\Psi')$ exists, we can construct a finite-dimensional compression $(C_A,C_B,\K,\Psi)$ with $\dim\K'\ge\dim\K$ that satisfies the following properties:
    \begin{enumerate}[(i)]
        \item\label{it:tensor_split} 
            $\K = \K_A\ox \K_B$ and $C_A(\A_A)''=\B(\K_A)\ox\1$ and $C_B(\A_B)''=\1\ox\B(\K_B)$.
            Hence, there are unital completely positive maps $D_j:\A_j\to\B(\K_j)$ so that $C_A = D_A(\placeholder)\ox\1$ and similarly for $C_B$.
        \item\label{it:cyclicity}
            $[(D_A(\A_A)\ox\1)\Psi]=[(\1\ox D_B(\A_B))\Psi]=\K_A\ox\K_B$.
    \end{enumerate}
    Any compression with respect $\omega$ that satisfies these two properties also satisfies:
    \begin{enumerate}[resume*]
        \item\label{it:schmidty}
            $\dim\K_A=\dim\K_B$ and $\Psi\in \K_A\ox\K_B$ is fully entangled in the sense that its Schmidt rank is $\dim\K_j$.
            Therefore, $(D_A,D_B,\Psi)$ satisfy the criteria of \cref{it:op_split}.
        \item\label{it:stinespring_irrep}
            The minimal Stinespring dilation $(\pi_j,\H_j,W_j)$ of $D_j:\A_j\to\B(\K_j)$ is irreducible, $j=A,B$.
            Therefore, $(\pi_A,\pi_B,\Omega)$, with $\Omega=W_A\ox W_B\Psi$, satisfy the properties of \cref{it:normal_form} of \cref{thm:pure_tame_states}.
        \item\label{it:schmidt}
            $\SRdef{it:alg}(\omega)$ is equal to the vector Schmidt rank of $\Psi$ (and hence finite).
    \end{enumerate}
\end{prop}

We will prove later that \cref{it:cyclicity} already determines the compression up to unitary equivalence (see \cref{sec:uniqueness}).
Recall that if $(C_A,C_B,\K,\Psi)$ is a compression with respect to a bipartite state $\omega$, then we have a unital completely positive map $C:\A_{max}\to\B(\K)$ whose marginals are the maps $C_j$ (see the discussion around \eqref{eq:compression_C}).
For the proof, we start with a small Lemma:

\begin{lem}\label{thm:lemma_finite_compression}
Let $(C_A,C_B,\K,\Psi)$ be a finite-dimensional compression such that $C(\A_{max})$ acts irreducibly on $\K$.
Consider the $^*$-algebras $\N_j = C_j(\A_j)''$. Then $\N_A$ and $\N_B$ are factors and $\N_A = \N_B'$.
Therefore $\K = \K_A \otimes \K_B$ and $C_A=D_A\ox \1$, $C_B=\1\ox D_B$.
\end{lem}
\begin{proof}[Proof of the Lemma]
    We start by showing $C(\A_{max}) \subset \N_A\vee \N_B$.
    It suffices to prove $C(x)\in \N_A\vee \N_B$ for elements $x=\sum_{i=1}^n a_i\ox b_i$ of the algebraic tensor product.
    We have $C(x) = \sum_i C_A(a_i)C_B(b_i) \in \N_A\vee \N_B$. 
    It now also follows that $\N_A\vee \N_B \supset C(\A_{max})''=\B(\K)$ so that we have $\N_A\vee \N_B = \B(\K)$.
    Furthermore, commutativity of $C_A(\A_A)$ and $C_A(\A_B)$ implies that $\N_A$ and $\N_B$ commute which by finite-dimensionality (and $\N_A\vee \N_B=\B(\K)$) implies $\N_A=\N_B'$.
    Therefore $\K=\K_A\ox\K_B$ and $\N_A=\B(\K_A)\ox\1$, $\N_B=\1\ox\B(\K_B)$. 
\end{proof}

\begin{proof}[Proof of \cref{thm:finite_compression}]
\ref{it:tensor_split}: 
This follows from \cref{thm:lemma_finite_compression} if we show that we may assume that $C(\A)$ has no invariant subspaces.
Set $\M = C(\A_{max})''$. As a finite-dimensional $C^*$-algebra, $\M$ is a direct sum of matrix algebras so that we have
Then we have direct sum decompositions $\K = \bigoplus_\gamma \K_\gamma$, $\M=\bigoplus_\gamma\B(\K_\gamma)$, $C(x) = \oplus_\gamma C_\gamma(x)$ and $\Psi = \oplus \sqrt{p_\gamma} \Psi_\gamma$ with $\norm{\Psi_\gamma}=1$.
We can discard any direct summands with $p_\gamma=0$.
Since $\omega$ is pure, it follows that $(C_\gamma|_{\A_A},C_\gamma|_{\A_B},\K_\gamma,\Psi_\gamma)$ are themselves compressions of $\omega$ each of which satisfies $C_\gamma(\A)''=\B(\K_\gamma)$.
We may pick any one of these and apply \cref{thm:lemma_finite_compression}.

\ref{it:cyclicity}:
Let $P_A$ be the projection onto the closed subspace $D_A(\A_A)\ox\1\Psi$. Then  $P_A \in  C_A(\A_A)' = \1\ox \B(\K_B)$ showing that $P_A = \1\ox Q_B$ for a projection $Q_B$ on $\K_B$.
Similarly, the projection $P_B$ onto  $\1\ox D_B(\A_B)\Psi$ is of the form $Q_A\ox \1$.
We can now simply truncate everything with these projections, i.e.\ replace $\K_j$ by $Q_j\K_j$ and $D_j$ by $Q_jD_j(\placeholder)Q_j$ (we already have $\Psi = Q_A\ox Q_B\Psi$).

\ref{it:schmidty}: This follows from \ref{it:cyclicity} Let $\Psi = \sum_{\alpha=1}^k \lambda_\alpha\Phi_\alpha^A\ox \Phi_\alpha^B$, $k\in\NN\cup\{\oo\}$, be the Schmidt decomposition. 
Clearly, $\1\ox D_B(\A_B)\Psi \subset \CC^k\ox \lin\,\{\Phi^B_1,\ldots,\Phi^B_k\}$.
It is therefore necessary that $\{\Phi_\alpha^B\}_{\alpha=1}^k$ span $\K_B$ as we would have a contradiction to \ref{it:cyclicity} otherwise.

\ref{it:stinespring_irrep}:
By the Radon-Nikodym theorem, the claim is equivalent to: All completely positive maps $T:\A_j \to \B(\K_j)$ with $T\le D_j$ are proportional to $D_j$, i.e.\ $D_j$ is extremal, for $j=A,B$.
Let $T:\A\to\B(\K_A)$ be a completely positive map that is cp-dominated by $D_A$. 
It follows that the positive linear functional $\ip\Psi{(T\otimes D_B)(\placeholder)\Psi}$ is dominated by $\omega$ which by purity implies that there is a $\lambda>0$ such that $\ip\Psi{T(a)\otimes D_B(b)\Psi} = \lambda \omega(ab)$ for all $(a,b)\in\A_A\times\A_B$.
Let $\Psi=\sum_{i=1}\lambda_i\Phi_i^A\ox\Phi_i^B$, $k\in\NN\cup\{\oo\}$, be the Schmidt decomposition of $\Psi$.
Since $\1\otimes D_B(\A_B)\Psi = \K_A\otimes\K_B$, we can pick $b_{ij}\in \A_B$ so that $D_B(b_{ij})\Psi = \Phi_i^A\otimes \Phi_j^B$.
Then 
\begin{align*}
    \ip{\Phi_i^A}{T(a)\Phi_j^A} 
    &= \sum_k \frac{\lambda_k}{\lambda_i} \ip{\Phi^A_k}{D(a)\Phi_j^A} \ip{\Phi_k^B}{\Phi_i^B}\\
    &=\lambda_i^{-1} \ip\Psi{(T(a)\Phi_j^A)\otimes \Phi_k^B}\\
    &=\lambda_i^{-1} \ip\Psi{D_A(a)\otimes D_B(b_{ij})\Psi} =\lambda_i \omega(ab_{ij}).
\end{align*}
This shows that the matrix elements of $T(a)$, $a\in\A_A$, are fully determined.
In particular, we get $T(a) = \lambda D_A(a)$ for all $a\in\A_A$ as $D_A$ is also cp-dominated by $D_A$.
The same argument shows the claim for the second system.

\ref{it:schmidt}:
    As already noted in \cref{it:stinespring_irrep}, this implies that $(\pi_A,\pi_B,\Omega)$ satisfy the properties of \cref{it:normal_form} of \cref{thm:pure_tame_states}.
    Since local isometries do not change the Schmidt rank, we know that $\Omega=W_A\ox W_B\Psi$ has Schmidt rank $k$.
    From \cref{thm:tame_GNS}, we know that $\pi_\omega(ab)=\pi_A(a)\ox\pi_B(b)$ and $\Omega_\omega=\Omega$.
    Therefore, the GNS vector has Schmidt rank $k$ with respect to the tensor splitting of the GNS space.
\end{proof}

We collect all direct consequences that this has for the claimed equivalence of Definitions \ref{it:alg}, \ref{it:op_cof} and \ref{it:op_split}: 

\begin{cor}
    \begin{enumerate}[(1)]
        \item\label{it:1} 
            \cref{it:op_cof,it:op_split} are equivalent.
        \item\label{it:2} 
            If a pure bipartite state admits a finite-dimensional compression, then it is tame. Therefore, $\SRdef{it:op_cof}(\omega)=\SRdef{it:op_split}(\omega)=\oo$ for all wild bipartite pure states.
        \item\label{it:3}
            Definitions \ref{it:alg}, \ref{it:op_cof} and \ref{it:op_split} agree on bipartite pure states $\omega$ so that $\SRdef{it:op_split}(\omega)<\oo$.
        \item\label{it:4}
            Definitions \ref{it:alg}, \ref{it:op_cof}, and \ref{it:op_split} agree on wild bipartite pure states. 
    \end{enumerate} 
\end{cor}

\begin{proof}
    \ref{it:1} is immediate from \cref{it:tensor_split,it:cyclicity,it:schmidty} of \cref{thm:finite_compression}.
    By \cref{thm:pure_tame_states}, \ref{it:2} follows from \cref{it:stinespring_irrep} of \cref{thm:finite_compression}.
    \ref{it:3} is proved in \cref{it:schmidt} of \cref{thm:finite_compression}.
    \ref{it:4} is is a direct consequence \ref{it:1} and \ref{it:2}. 
\end{proof}

The remaining step in the proof is:

\begin{lem}
    Definitions \ref{it:alg}, \ref{it:op_cof} and \ref{it:op_split} agree on tame bipartite pure states $\omega$ so that $\SRdef{it:op_split}(\omega)=\oo$.
\end{lem}
\begin{proof}
    We need to prove that $\SRdef{it:op_split}=\oo$ implies $\SRdef{it:alg}=\oo$.
    We do this by showing the contrapositive: $\SRdef{it:alg}<\oo$ implies $\SRdef{it:op_split}<\oo$:
    So, assume $k:=\SRdef{it:alg}(\omega)<\oo$.
    Let $\pi_j:\A_j\to\B(\H_j)$, and $\Omega\in\H_A\ox\H_B$ be as in \cref{thm:pure_tame_states} so that $\ip\Omega{\pi_A(a)\ox\pi_B(b)\Omega}=\omega(ab)$ for all  $(a,b)\in\A_A\times\A_B$. 
    Let $\Omega=\sum_{\alpha=1}^k\lambda_\alpha \Phi^A_\alpha\ox\Phi^B_\alpha$ be the Schmidt decomposition, set $\K_j=\lin\{\Phi_1^j,\ldots,\Phi_k^j\}\equiv\CC^k$ and let $P_j$ be the orthogonal projection onto $\K_j$.
    Define $C_j:= P_j\pi_j(\placeholder)P_j:\A_j\to\B(\K_j)\equiv\MM_k$ and note that $\Omega\in\K_A\ox\K_B\subset\H_A\ox\H_B$. 
    Then $(C_A,C_B,\Omega)$ satisfy the criteria of \cref{it:op_split}, so that we indeed find $\SRdef{it:op_split}(\omega)\le k <\oo$.
\end{proof}

Combining these results, we see that the Definitions \ref{it:alg}, \ref{it:op_cof}, and \ref{it:op_split} agree on all bipartite pure states.
\qed

\subsubsection*{Step 2. Equivalence of \cref{it:op_split,it:factor}}

If the Schmidt rank in the sense of \cref{it:op_split} is finite, there are unital completely positive maps $C_j:\A_j\to\MM_k$ and a $\Psi\in\CC^k\ox\CC^k$ such that $\omega(ab)=\psi(C_A(a)\ox C_B(b))$, where $\psi$ is the pure state on $\MM_k\ox\MM_k$ implemented by $\Psi$.
$\psi$ corresponds to a completely positive map $\Gamma_\psi : \MM_k\to\MM_k^*$ as in \cref{eq:state_cp}.
We define completely positive maps $\alpha$ and $\beta$ by $\alpha = C_A$ and $\beta = C_B^*\circ\Gamma_\psi$.
These indeed factor $\Gamma_\omega$ through $\MM_k$: $\beta(\alpha(a))(b)= \Gamma_\psi(C_A(a))(C_B(b))=\psi(C_A(a)C_B(b)) = \omega(ab)$.
This is illustrated in the following commuting diagram 
\begin{equation}
    \begin{tikzcd}
        \A_A \arrow{rr}{\Gamma_\omega} \arrow{d}{C_A} 
        &&\A_B^*\\
        \MM_k\arrow{rr}{\Gamma_\psi}\arrow{urr}{\beta} &&\MM_k^*\arrow{u}{C_B^*}
    \end{tikzcd}
\end{equation}
This shows $\SRdef{it:factor}(\omega) \le \SRdef{it:op_split}(\omega)$.

Let $\alpha$ and $\beta$ be completely positive maps factorizing $\Gamma_\omega$ through $\MM_k$.
By \cref{thm:factorization}, we may assume $\alpha$ to be unital and $\beta(\1)$ to be a state.
Therefore, there is a state $\varphi$ on $\MM_k\ox \A_B$ defined by $\varphi(X\ox b) = \beta(X)(b)$.
By construction, it then holds that $\omega(ab) = \varphi(\alpha(a)\ox b)$.
By \cref{thm:one_sys_finite}, there is a unital completely positive map $C_B:\A_B\to\MM_k$ and a unit vector $\Phi\in\CC^k\ox\CC^k$ such that $\varphi(X\ox b)=\ip\Phi{X\ox C_B(b)\Phi}$.
Actually in \cref{thm:one_sys_finite} the map $C_B$ maps to $\MM_r$, with $r\le k$ being the Schmidt rank of the vector $\Phi$, but we can just take a direct sum with some arbitrary unital completely positive map to ensure what we claimed.
If we now set $C_A=\alpha$, the triple $(C_A,C_B,\Phi)$ satisfies the requirements in \cref{it:op_split} because $\ip\Phi{C_A(a)\ox C_B(b)\Phi}=\varphi(\alpha(a)\ox b) = \omega(ab)$.
Since this works for all $k$ that admit a factorization of $\Gamma_\omega$ through $\MM_k$, it follows that $\SRdef{it:factor}(\omega)\ge \SRdef{it:op_split}(\omega)$.
\qed

\subsubsection*{Step 3. Equivalence of \cref{it:alg,it:rank,it:interval}}

In particular, we have to show the equalities $\rank(\omega_A)=\rank(\omega_B)$ and $\dim\,[0,\omega_A]=\dim\,[0,\omega_B]$.
Otherwise, \cref{it:rank,it:interval} are not even well-defined.
We start with the following Lemma which connects the rank of a state (see \cref{eq:rank}) to properties of the GNS representation.

\begin{lem}\label{thm:rank}
    Let $\varphi$ be a factor state on a $C^*$-algebra $\B$.
    \begin{enumerate}[(1)]
        \item 
            If $\pi_\varphi(\B)$ is type II or III, then  $\rank(\varphi)=\oo$.
            If $\pi_\varphi(\B)$ has type I, then  $\pi_\varphi(\B)'$ has type I$_{\rank(\varphi)}$. 
        \item 
            $\rank\varphi = \sqrt{\dim\,[0,\varphi]}$. 
    \end{enumerate} 
\end{lem}

Recall that marginals of pure bipartite states are factor states (see \cref{thm:pure_tame_states_lemma}).
 
\begin{proof}
    Set $\M=\pi_\varphi(\B)'$.
    Denote by $\beta: \M \to \B^*$ the linear extension of the Radon-Nikodym bijection between $[0,\1]_\M$ and $[0,\varphi]$ (cf.\ beginning of \cref{sec:proof}).
    If $\varphi= \sum_{i=1}^k p_i\psi_i$ for pure states $\psi_i$, then $p_i\psi_i\in[0,\varphi]$. 
    Set $P_i = \beta^{-1}(p_i\psi_i)\in\M$.
    Then $0\le P_i\le \1$ and $\sum_{i=1}^k P_i =\1$.
    Sine $\psi_i$ is pure, $P_i$ is a projection (because the projections are the extremal points of the unit interval) and all positive operators in $\M$ which are dominated by $P_i$ are proportional to $P_i$.
    Therefore all $P_i\M P_i \cong \CC$ so that the $P_i$ are orthogonal one-dimensional projections.
    Since one-dimensional projections only exist in type I algebras, the first claim of the first item follows.
    Since $\M$ is a factor and we can write the identity of $\M$ as a sum on $k$ minimal projections, $\M$ has type I$_k$ which proves the first item.
    The second item follows because 
    \[
        \dim\,[0,\varphi] := \dim(\lin\,[0,\varphi]) = \dim \M.
    \] 
    The vector space dimension of a factor is infinite, except if it is a type I$_n$ factor in which case $\dim\M = n^2$. 
\end{proof}

\begin{lem}
    Let $\omega$ be a bipartite pure state. Set $k=\SRdef{it:alg}(\omega)\in\NN\cup\{\oo\}$. If $\omega$ is a tame state, then $\pi_{\omega_A}(\A_A)'$ and $\pi_{\omega_B}(\A_B)'$ are type I$_k$ factors.
\end{lem}

\begin{proof}
    This is immediate from the explicit construction of the GNS representation of marginals of tame bipartite pure states in \cref{it:tame_GNS_marginal} of \cref{thm:tame_GNS}.
\end{proof}

If we combine these two Lemmas, the equivalence of \cref{it:alg,it:interval,it:rank} follows.
\qed

\subsection{Uniqueness of minimal compressions}\label{sec:uniqueness}

If a state satisfies Haag-duality, then we can show that there is a unique minimal compression: 

\begin{prop}\label{thm:uniqueness}
    If $\omega$ satisfies Haag-duality, a unique compression $(C_A,C_B,\K,\Psi)$ with respect to $\omega$ satisfying $[C_j(\A_j)\Psi]=\K$, $j=A,B$, and $\sqrt{\dim\K}=\SR(\omega)$ exists.
\end{prop}

The uniqueness here means uniqueness up to unitary equivalence with local unitaries.

\begin{proof}
    Existence:
    Denote by $Q_j$ the projection onto $\V_j:=[\pi_\omega(\A_j)\Omega_\omega]\subset \H_\omega$.
    By Haag-duality $Q_A\in \M_A'=\M_B$ and $Q_B\in\M_B'=\M_A$ so that $Q_A$ and $Q_B$ commute.
    Thus $Q=Q_AQ_B$ is an orthogonal projection and $Q\H_\omega=\V_A\cap\V_B$.
    Define $C = Q\pi_\omega(\placeholder)Q$, $C_j=C|_{\A_j}$, $\K = \V_A\cap\V_B$ and $\Psi=\Omega_\omega$.
    It follows that $C_A(a)C_B(b) = Q\pi_\omega(a)Q_AQ_B\pi_\omega(b)Q=QQ_A\pi_\omega(ab)Q_BQ=C(ab)$ which implies $C_A(a)C_B(b) = C(b^*a^*)^*= C(a^*b^*)^*=C_B(b)C_A(a)$ and $\ip\Psi{C_A(a)C_B(b)\Psi}=\ip{\Omega_\omega}{\pi_\omega(ab)\Omega_\omega}=\omega(ab)$.
    By construction, we have $[C_j(\A_j)\Psi]=[Q_BQ_A\pi_\omega(\A_j)\Omega_\omega]=[Q_BQ_A\V_j]=\K$.

    Uniqueness: 
    Consider the unital completely positive map $C:\A_{max}\to\B(\K)$ whose marginals are $C_A$ and $C_B$.
    Let $(\tilde\pi,\tilde\K,W)$ be the minimal Stinespring dilation of the unital completely positive map $C:\A_{max}\to \B(\K)$ and set $P=W^*W$.
    Define $\Omega=W\Psi$ and $\V_j=[\tilde\pi(\A_j)\Omega]$.
    The subspace $\H=[\tilde\pi(\A)\Omega]$ contains $\V_A$ and $\V_B$ and the restriction $\pi=\tilde\pi(\placeholder)|_{\H}$ is the GNS representation of $\omega$.
    The equation $C(a)C(b)=C_A(a)C_B(b)=C(ab)$ is equivalent to $P\tilde\pi(a)P\tilde\pi(b)P=P\pi(ab)P$ and implies
    \begin{equation*}\label{eq:help2}
       \ip\Psi{P\Phi}=\ip\Psi\Phi\qquad\forall(\Psi,\Phi)\in\V_A\times\V_B. 
       \tag{$\triangle$}
    \end{equation*}
    Furthermore, $[C_j(\A_j)\Psi]=\K$, $j=A,B$, implies
    \begin{equation*}\label{eq:help3}
        P\tilde\K = [P_j\V_j], \quad j=A,B.  
        \tag{{\raisebox{0.2ex}{\footnotesize$\bigcirc$}}}
    \end{equation*}
    To show the uniqueness it suffices to show that $P$ is the projection onto $\V_A\cap\V_B$ since this is how we constructed the compression in the first part of the proof.
    We claim that $P$ acts as the identity on $\V_A\cap \V_B$.
    This is seen from \eqref{eq:help2} $\norm{P\Psi}^2=\ip\Psi{P\Psi}=\ip\Psi\Psi=\norm\Psi^2$ for $\Psi\in\V_A\cap\V_B$.
    We now assume that $P\Psi=\Psi$ and that $\Psi \perp \H_A\cap\H_B$ and have to show that this implies $\Psi=0$.
    Since $\omega$ satisfies Haag duality the projection $Q_A$ onto $\V_A$ and $Q_B$ onto $\V_B$ commute as projections on $\H$ and hence also as projections on the potentially larger space $\tilde\K$.
    This commutativity implies that 
    $$
    \tilde\K=\V_A\cap\V_B\oplus\V_A\cap\V_B^\perp \oplus\V_A^\perp\cap\V_B\oplus \V_A^\perp\cap\V_B^\perp.
    $$
    With respect to this, we have $\Psi=0+\Psi_A+\Psi_B+\Psi_\perp$. 
    By \eqref{eq:help3}, there are $\Phi_j\in\V_j$ such that $\Psi=P\Phi_A=P\Phi_B$.
    Therefore, $\norm{\Psi_A}^2=\ip{\Psi_A}{\Psi}= \ip{\Psi_A}{P\Phi_B}=\ip{\Phi_A}{\Phi_B}=0$ where we used \eqref{eq:help2}.
    Similarly, one sees $\Phi_B=0$.
    Consequently, $\Psi=\Psi_\perp\in\H_A^\perp\cap\H_B^\perp$ and it follows that $\norm{\Psi}^2=\ip{P\Psi}{\Phi_A} = \ip{\Psi}{\Phi_A}=0$.
    This proves that $P\tilde\K = \V_A\cap\V_B\subset\H$ and hence, that $C= P\tilde\pi(\placeholder)P = Q_AQ_B\pi(\placeholder)Q_AQ_B$ where $Q_j$ is the projection onto $\V_j$ and $\pi = \tilde\pi(\placeholder)|_\H$.
\end{proof}

\begin{cor}
    Every tame pure state has a unique compression of the form $(C_A\ox\1,C_B\ox\1,\K_A\ox\K_B,\Psi)$ satisfying $[C(\A_A)\ox\1\Psi]=[\1\ox C(\A_B)\Psi]=\K_A\ox\K_B$ and $\dim\K_A=\dim\K_B=\SR(\omega)$.
    If $\SR(\omega)<\oo$ then every compression with $\sqrt{\dim\K}=k$ is of this form.
    It can be obtained from the factorization of the GNS representation \cref{thm:tame_GNS} by restricting the local Hilbert spaces to the local supports of the GNS vector.
\end{cor}

\subsection{Properties of the Schmidt rank}\label{sec:properties}

We collect some properties that are enjoyed by the Schmidt rank. These are lower-semicontinuity, multiplicativity under tensor products, and monotonicity under local operations.

\begin{prop}
    The Schmidt rank is $w^*$-lower-semicontinuous on the set of pure states on a fixed bipartite algebra $\A$, i.e.\ if $\omega_\alpha$ is a net of pure states on $\A$ $w^*$-converging to a pure state, then
    \begin{equation}\label{eq:lsc}
        \liminf_\alpha \, \SR(\omega_\alpha) \ge \SR(\lim_\alpha\omega_\alpha).
    \end{equation}
\end{prop}
\begin{proof}
    We may assume $k:=\liminf_\alpha\SR(\omega_\alpha)<\oo$ and, by passing to a subnet, that $\SR(\omega_\alpha)\le k$ for all $\alpha$.
    Consider the set $K$ of triples $(C_A,C_B,\Psi)$ with unital completely positive maps $C_j:\A_j\to\MM_k$ and unit vector $\Psi\in\CC^k\ox\CC^k$.
    Equipped with the product topology of the topology of pointwise convergence for the maps $C_A$ and the standard topology for the vector, $K$ is a compact set, and the map $E:K \to \states(\A_{max})$ is continuous with respect to the $w^*$-topology.
    Therefore the image $E(K)$ is $w^*$-closed and hence contains the state $\lim_\alpha\omega_\alpha$ if $\SR(\omega_\alpha)\le k$. 
    Since the limit is pure by assumption, the result follows.
\end{proof}

\begin{prop}\label{thm:SR_is_multiplicative}
    The Schmidt rank is multiplicative under tensor products:
    Let $\omega$ and $\varphi$ be pure states on bipartite algebras $\A$ and $\B$.
    Then $\omega\ox\varphi$ is a pure state on the bipartite algebra $\A\ox_{min}\B$ (with the local algebras being $\A_j\ox_{min}\B_j$, $j=A,B$) and $\SR(\omega\ox\varphi)=\SR(\omega)\cdot \SR(\varphi)$.
\end{prop}

\begin{proof}
    This follows from the algebraic definition of the Schmidt rank.
    Since it holds that $\pi_{\omega\ox\varphi}=\pi_\omega\ox\pi_\varphi$, the product can only be tame if both $\omega$ and $\varphi$ are.
    The result follows because the vector Schmidt rank is multiplicative on tensor products of bipartite Hilbert spaces.
\end{proof}

Recall the definition of local operations (see \cref{def:local_op}).

\begin{prop}
    The Schmidt rank is monotone under local operations. 
    Let $\A$ and $\B$ be bipartite algebras and let $\omega$ and $\varphi$ be bipartite pure states on $\A$ and $\B$ respectively.
    If there is a local operation $T:\A\to\B$ such that $T^*(\varphi)=\omega$, then $\SR(\varphi)\ge\SR(\omega)$.
\end{prop}

\begin{proof}
    Define $T_j:\A_j\to\B_j$, $j=A,B$, as the marginal maps of $T$. Let $(C_A,C_B,\K,\Psi)$ be a compression of $\omega$.
    Define $D_j=C_j\circ T_j$, then $(D_A,D_B,\K,\Psi)$ is a compression of $\omega\circ T$ so that the inequality holds by \cref{it:op_cof} of the Schmidt rank.
\end{proof}

As an application, we consider entanglement distillation which allows us to obtain lower bounds on the Schmidt rank:

\begin{defin}
    Let $\omega$ be a bipartite state for local algebras $\A_A$ and $\A_B$.
    Let $\K_A$ and $\K_B$ be Hilbert spaces and let $\Psi\in\K_A\ox\K_B$ be a unit vector.
    We say that $\Psi$ is {\bf distillable} from $\omega$, if there are unital completely positive maps $D_j:\B(\K_j)\to\A_j$ such that 
    \begin{equation}\label{eq:distillation}
        \omega(D_A(a)D_B(b))=\ip\Psi{a\ox b \Psi} \qquad\forall (a,b)\in\B(\K_A)\times\B(\K_B).
    \end{equation}
\end{defin}

This definition is equivalent to the existence of a local operation from $\A_{max}$ to $\B(\K_A\ox\K_B)$ which takes $\omega$ to the bipartite state $\psi = \ip\Psi{(\placeholder)\Psi}$.

\begin{cor}
    The Schmidt rank of a distillable vector is a lower bound to the Schmidt rank of $\omega$, i.e.\ if $\Phi\in \K_A\ox\K_B$ is distillable from $\omega$, then $\SR(\omega)\ge\SR(\Phi)$.
\end{cor}

\begin{proof}
    If $\SR(\omega)=k$, then there are completely positive unital maps $C_j:\A_j\to\MM_k$ and a Schmidt rank-$k$ unit vector $\Psi\in\CC^k\ox\CC^k$ so that $\omega(ab)=\ip\Psi{C_A(a)\ox C_B(b)\Psi}$.
    Composition with the distillation maps $D_j:\B(\K_j)\to\A_j$ gives us unital completely positive maps $T_j=C_j\circ D_j:\B(\K_j)\to\MM_k$ which are able to distill $\Phi$ from a state of Schmidt rank $k$.
\end{proof}

In \cite{infinitely_entangled_states}, bipartite states are said to have infinite distillable entanglement if it is possible to distill vectors of arbitrary Schmidt rank.

\begin{cor}
    States with infinite distillable entanglement have infinite Schmidt rank.
\end{cor}

\section{Examples and applications}\label{sec:examples}

\subsection{Gapped ground states and finitely correlated states on spin chains}\label{sec:spinchains}

We consider one-dimensional spin chains. We regard these as bipartite systems between the left and right sides.

The full $C^*$-algebra is $\A_{\ZZ}=\bigotimes_{x\in\ZZ} \MM_d$ (formally, the infinite tensor product is defined as an inductive limit).
To a region $\Lambda \subset \ZZ$, one associates the algebra $\A_\Lambda = \bigotimes_{x\in\Lambda}\MM_d$ which is identified with a subalgebra of $\A_\ZZ$ by tensoring with the identity on all sites in the complement $\Lambda^c$.
In particular, we consider the left and right chain algebras $\A_L = \A_{(-\oo,0]}$, $\A_R = \A_{[1,\oo)}$.

\subsubsection*{Gapped ground states}

It follows from a theorem by Matsui \cite{matsui} that gapped ground states on one-dimen\-sional spin chains are tame bipartite states if the interactions are not too crazy.
What Matsui really proves is that the state has the split property if it satisfies an area law that was proved to hold by Hastings in \cite{Hastings2007} in the following context:

An interaction is a map $\Phi$ from finite subsets $X\subset\ZZ$ to hermitian elements of $\A_X$ \cite{bratteli2}. 
To every finite subset $\Lambda \subset \ZZ$, we associate the Hamiltonian $H_\Lambda = \sum_{X\subset\Lambda} \Phi(X) \in \A_\Lambda$.
For simplicity, we restrict ourselves to uniformly bounded interactions of finite-range, i.e.\ $\sup\norm{\Phi(X)} <\oo$ and if there is an $N$ such that $\Phi(X)=0$ if  $\abs X>0$.
For such interactions, the dynamics generated by $H_\Lambda$ converge as $\Lambda\to\ZZ$ to the strongly continuous one-parameter automorphism group generated by the closure of the derivation
\begin{equation}
    \delta: \bigcup_n \A_{[-n,n]} \to \A_{\ZZ}, \ \delta(a) = i \lim_\Lambda [H_\Lambda,a],
\end{equation}
see for example \cite[Prop.~6.2.3]{bratteli2}.
For the precise definition of a gapped ground state, we refer to \cite{ogata2021classification} and the references therein.
What's important for us is that these are necessarily pure states on $\A_\ZZ$.

\begin{prop}[Matsui, Hastings]
    Let $\omega_\Phi$ be a gapped ground state of a uniformly bounded finite-range interaction. 
    If we consider $\omega_\Phi$ as a bipartite state for the left and right chain algebras $\A_L$ and $\A_R$, then $\omega_\Phi$ is a tame bipartite pure state.
\end{prop}

\subsubsection*{Finitely correlated states}

We now consider translation invariant states on infinite spin chains, which we again regard as bipartite for left and right sides.
We will see that we can explicitly determine the Schmidt rank for the Heisenberg anti-ferromagnet and the AKLT model by using the theory of finitely correlated states introduced in \cite{FNW}.

In \cite{FNW}, a translation-invariant state $\omega$ on $\A_\ZZ$ is said to be finitely correlated if the vector space 
\begin{equation}\label{eq:VR}
    \V = \set{\omega(a_L\ox \placeholder) \given a_L\in\A_L} \subset \A_R^*
\end{equation}
is finite-dimensional.
$\V$ can be equipped with the structure of an operator system whose order unit is $e = \omega(1\ox\placeholder)$ and which carries a natural state $\rho:\V\to\CC$ defined by $\omega(a\ox\placeholder)\mapsto\omega(a\ox1)$.%
\footnote{The matrix order and $^*$-operation on $\V$ are inherited from $\A_R^*$. That this indeed turns $\V$ into an operator system can, e.g.\ be seen from \cref{thm:ucoi}. We do not know if this matrix order is equivalent to the one constructed in \cite{FNW} by means of a quotient construction but we know that $\EE$ is completely positive with respect to both matrix orders.}
Furthermore, there is a unital completely positive map $\EE:\MM_d\ox \V\to \V$ sending $a\ox\omega(a_L\ox\placeholder)$ to $\omega(a_L\ox a \ox \placeholder)$.
We use the notation $\EE_{a}:=\EE(a\ox \placeholder)$, $a\in\MM_d$.
Together, $(\V,\EE,\rho)$ fully describe the state $\omega$: 
\begin{equation}\label{eq:finitely_corr}
    \omega(a_{-N}\ox a_{1-N}\ox \ldots a_N) = \rho(\EE_{a_{M}}\circ \EE_{a_{1-N}}\circ\ldots \EE_{a_N}(e)),
\end{equation}
where $a_i\in\MM_d$, $N\in\NN$, and where all other tensor factors are the identity element.
Conversely, every collection $(\S,\EE,\rho)$ of a (finite-dimensional) operator system $\S$ a unital completely map $\EE:\MM_d\ox\S\to\S$ and a state $\rho$ so that $\rho\circ\EE_{\1}=\rho$ defines a translation invariant (finitely correlated) state.
A translation-invariant state $\omega$ is called \emph{$C^*$-finitely correlated} if $\S$ can be chosen to be a finite-dimensional $C^*$-algebra.
In this case, one can always choose $\S = \MM_n$ (implying $e=\1_n$) and $\rho$ to be faithful, which we call an $n$-dimensional representation as a $C^*$-finitely correlated state.

\begin{thm}\label{thm:finitely_corr}
    Let $\omega$ be a translation invariant pure state on $\A_{\ZZ}$.
    The following are equivalent:
    \begin{enumerate}[(a)]
        \item\label{it:fcs_csfcs} $\omega$ is a $C^*$-finitely correlated state,
        \item\label{it:fcs_finiteSR} $\omega$ has finite Schmidt rank if regarded as a bipartite state for the left and right side,
        \item\label{it:fcs_fcs} $\omega$ is finitely-correlated and satisfies Haag-duality.
    \end{enumerate} 
    In this case, the Schmidt rank is the smallest dimension $n$ so that $\omega$ allows an $n$-dimensional representation as a $C^*$-finitely correlated state.
\end{thm}

From the applications of finitely correlated states to the AKLT model and the Heisenberg model in \cite[Ex.~1 and Cor.~7.3]{FNW}, we get:

\begin{cor}
    \begin{enumerate}[(1)]
        \item 
            The ground state of the generalized AKLT model has Schmidt rank $2$.
        \item 
            The ground state of the Heisenberg anti-ferromagnet has an infinite Schmidt rank. 
    \end{enumerate} 
\end{cor}

Another consequence of \cref{thm:finitely_corr} is:

\begin{cor}\label{thm:FCS_but_not_CSFCS}
    Let $\omega$ be a pure translation invariant state.
    If $\omega$ is finitely correlated but not $C^*$-finitely correlated, then $\omega$ does not satisfy Haag duality.
\end{cor}

Therefore such a state $\omega$ is necessarily wild and has an infinite Schmidt rank. 
To prove \cref{thm:finitely_corr}, we need the following Lemma, which might be of independent interest (see \cref{sec:outlook}).

\begin{lem}\label{thm:Haag-duality_order_int}
    Let $\omega$ be a pure bipartite state for algebras $\A_A$ and $\A_B$.
    If $\omega$ satisfies Haag-duality, then the linear functionals $a \mapsto \omega(ab)$, $b\in\A_B$ are $w^*$-dense in the linear hull of the order interval $[0,\omega_A]\subset\A_A^*$. The same holds for the other system.
\end{lem}

\begin{proof}[Proof of \cref{thm:Haag-duality_order_int}]
    Consider the GNS representation $(\pi_\omega,\H_\omega,\Omega_\omega)$ of $\omega$ and the projection $P_A$ onto $[\pi_{\omega}(\A_A)\Omega]$.
    Then, by Haag-duality, $P_A \in \M_A'=\M_B$.
    We know that $(P_A \pi_\omega(\placeholder)P_A,P_A\H_\omega,\Omega_\omega)$ is the GNS representation of the reduced state.
    Then $\N_A = P_A\M_A P_A$ is the von Neumann algebra corresponding to the reduced state $\omega_A$.
    Now let $\nu\in\A_A^*$ with $0\le \nu\le \omega_A$ be given and let $Q\in \N_A'$ be such that $\nu = \ip{\Omega_\omega}{QP_A\pi(\placeholder)P_A\Omega_\omega}$ (this exists by \cref{thm:affine_bij}).
    Since the commutant is $\N_A' = P_A\M_A'P_A=P_A\M_BP_A$, there is a net $(b_n)_n$ in $\A_B$ so that $P_A \pi(b_n) P_A$ converges weakly to $Q$.
    This implies that $\omega(ab_n)$ converges to $\nu(a)$ for all $a\in\A_A$.
\end{proof}

\begin{proof}[Proof of \cref{thm:finitely_corr}]
    We use the notation $\omega_R$ for the reduced state onto $\A_R$, i.e.\ $\omega_R=\omega(1\ox\placeholder)$. 

    \ref{it:fcs_csfcs} $\Rightarrow$ \ref{it:fcs_finiteSR}:
    Let $n$ be minimal so that $(\MM_n,\EE,\rho)$ is a representation of $\omega$ as a $C^*$-finitely correlated state $\omega$. 
    We define a completely positive $\alpha:\A_R \to \MM_n$ by linear and continuous extension of
    \[
        \alpha(a_{1}\ox \ldots \ox a_N) = (\EE_{a_{1}} \circ \ldots\circ \EE_{a_N})(\1).
    \] 
    We now define a completely positive map $\beta:\MM_n\to\A_L^*$ by
    \[
        \beta(x)(a_{-N}\ox\ldots\ox a_{0}) = \rho((\EE_{a_{-N}}\circ\ldots\circ \EE_{a_0})(x)).
    \] 
    We claim that $\beta\circ\alpha=\Gamma_\omega$, which implies that the Schmidt rank of $\omega$ is less than $n$ (using \cref{def:factor}):
    \begin{align*}
        \beta(\alpha(a_{1}\ox \ldots a_N))(a_{-N}\ox \ldots a_0)
        &= \rho(\EE_{a_{-N}}\circ\ldots \EE_{a_0}\circ\EE_{a_1}\circ\ldots\EE_{a_N}(\1))\\
        &=\omega(a_{-N}\ox\ldots a_N)\\
        &=\Gamma_\omega(a_1\ox a_N)(a_{-N}\ox\ldots a_0).
    \end{align*}

    \ref{it:fcs_finiteSR} $\Rightarrow$ \ref{it:fcs_fcs}:
    Every bipartite state with finite Schmidt rank is tame and hence satisfies Haag duality.
    Furthermore, it is clear that 
    \begin{equation*}\label{eq:help1}
        \V=\set{\omega(a_L\ox \placeholder)\given a_L\in\A_L}\subset \lin([0,\omega_R])=:\W.\tag{$\square$}
    \end{equation*}
    Thus $\V$ is finite-dimensional because $\W$ is finite-dimensional ($\dim\ \W = \SR(\omega)^2$ by \cref{it:interval} in \cref{thm:schmidt_rank}). 

    \ref{it:fcs_fcs} $\Rightarrow$ \ref{it:fcs_csfcs}:
    Since $\omega$ satisfies Haag-duality, \cref{thm:Haag-duality_order_int} implies that the inclusion \eqref{eq:help1} is $w^*$-dense.
    Therefore the assumption that $\V$ is finite-dimensional implies that $\W$ contains a $w^*$-dense finite-dimensional subspace which forces $\V=\W$.
    The Radon-Nikodym theorem for states gives us a complete order isomorphism $\lambda:\pi_{\omega_R}(\A_R)'\to \W\subset\A^*$, i.e.\ $\lambda$ and $\lambda^{-1}$ are completely positive (see \cref{thm:ucoi} in \cref{sec:appendix}).
    As we know that $\V=\W$ (both equipped with the matrix order inherited from $\A^*$), this shows that $\V$ is completely order isomorphic to $\pi_{\omega_R}(\A)'$. 
    Under this isomorphism the unit $e$ of $\V$ is mapped to $\1\in\pi_{\omega_R}(\A_R)'$.
\end{proof}

\subsection{Purification in algebraic quantum mechanics}\label{sec:purification}

Purification is the phenomenon that every state of a quantum system can be obtained from a pure bipartite state by discarding an ancillary system, i.e.\ that every quantum state is the marginal of a pure state.
What makes purification work is that bipartite pure states can be entangled.
In fact, purification of a classical system is impossible since every pure bipartite state is necessarily a product state.

It is often said that the GNS representation is an operator-algebraic generalization of the purification.
However, the GNS representation does not provide a bipartite system, and it also exists for classical systems where purification is impossible.

\begin{prop}\label{thm:purification}
    Let $\omega$ be a state on a $C^*$-algebra $\A$.
    The following are equivalent
    \begin{enumerate}[(a)]
        \item $\omega$ is a factor state,
        \item $\omega$ admits a purification in the sense that there exists a $C^*$-algebra $\B$ and a bipartite pure state $\psi$ on $\A\ox_{max}\B$ so that $\omega = \psi|_\A$.
    \end{enumerate} 
    Furthermore, $\psi$ is tame if and only if $\omega$ is a type I factor state. In this case the Schmidt rank of $\psi$ is the number $k\in\NN\cup\{\oo\}$ such that $\pi_\psi(\A)'$ type I$_k$.
    If these hold, $\psi$ can always be chosen to satisfy Haag-duality. If $\A$ is separable, $\B$ can also be chosen to be separable. 
\end{prop}

We briefly sketch the proof:
We already know that marginals of pure states are always factor states by \cref{thm:pure_tame_states_lemma} so that a purification can only work for factor states anyhow.
If $\pi_\omega(\A)''$ is a factor then so is $\pi_\omega(\A)'$ and we can construct the algebra $\B$ from it and a state $\psi$:
Pick an $\sigma$-weakly dense subalgebra $\B\subset\pi_\omega(\A)'$ (separable if $\A$ is) and define an irreducible representation $\pi_\psi:\A\ox_{max}\B\to\B(\H_\omega)$ by $\pi_{\psi}(a\ox b)=\pi_{\omega}(a)b$ (as $\pi_{\omega}$ is factor state).
The pure state $\psi$ on $\A\ox_{max}\B$ is defined by $\psi = \ip{\Omega_\omega}{\pi_{\psi}(\placeholder)\Omega_\omega}$. 
Clearly $(\pi_\psi,\H_\omega,\Omega_\omega)$ is the GNS representation of $\psi$ so that $\psi$ is indeed pure (because $\pi_\psi$ is irreducible).
Finally we have $\psi(a\ox\1) = \ip{\Omega_\omega}{\pi_\psi(a)\Omega_\omega} = \ip{\Omega_\omega}{\pi_\omega(a)\Omega_\omega}=\omega(a)$.

\subsection{Araki-Woods-Powers states}\label{sec:awp}

The construction of factors of type III due to Powers \cite{powers1967uhf} and Araki-Woods \cite{araki1968classification} yields examples of tame and wild bipartite states that satisfy Haag-duality. 
As the bipartite algebra we consider the infinite $C^*$-tensor product $\A = \ox_{j\in\ZZ}\MM_{d_{j}}$ generated by the left and right $C^*$-subalgebras $\A_{A} = \ox_{k\in\ZZ\setminus\NN}\MM_{d_{-k}}$ and $\A_{B} = \ox_{k\in\NN}\MM_{d_{k}}$ corresponding to two separated physical systems. 
Here $d_k\in\NN$ are arbitrary integers such that $d_k=d_{1-k}$.
For every $k$, let $\lambda_1\up k\ge \lambda_2\up k\ge \ldots \ge \lambda_{d_k}\up k>0$ be decreasingly ordered numbers that sum to one and consider the unit vector
\begin{equation}
    \Omega_k = \sum_{r=1}^{d_k} (\lambda_r\up k)^{\frac12} \ket r\ox\ket r\in \CC^{d_{1-k}}\ox\CC^{d_k}.
\end{equation}
For the bipartite state $\omega$ on $\A$ we take the infinite tensor product state $\omega = \ox_{k\in\NN}\omega_{k}$, where $\omega_{k}$ are the pure states on $\MM_{d_{1-k}}\ox\MM_{d_{k}}$ implemented by vectors $\Omega_k$.
Both marginals of $\omega_k$ are equal to the state $\varphi_k([x_{ij}]) = \sum_{i=1}^{d_k} \lambda_i\up k x_{ii}$ on $\MM_{d_k}=\MM_{d_{1-k}}$.
As each $\omega_{k}$ is a pure state, it follows that $\omega$ is a pure state on $\A$. By construction the marginals $\omega_{A}$ and $\omega_{B}$ are both given by the infinite tensor product state $\ox_{k\in\NN}\varphi_k$, and the associated von Neumann algebras $\M_{A}$ and $\M_{B}$ satisfy Haag-duality $\M_{A} = \M_{B}'$ because the restrictions of the GNS representation $\pi_{\omega}$ to $\A_{A}$ and $\A_{\B}$ yields the standard form \cite{bratteli1}.
The type of the factor $\M_{A}$ (and therefore also $\M_{B}$ as each is anti-isomorphic to its commutant) can be decided from the asymptotic behavior of the Schmidt coefficients $\{\lambda^{(k)}_{r}\}_{r=1,\ldots,d_{k}}$ as $k\to\oo$. 
The central result in this respect is the following list of conditions \cite{araki1968classification} (see \cite[Ch.~5, Sec.~4]{connes1994ncg} and \cite[Ch.~3, Sec.~3]{blackadar2006oa} for further discussion):
\begin{itemize}
	\item[1.] $\M_{A}$ is type I if and only if $\sum_{k=1}^\oo|1-\lambda^{(k)}_{1}|<\oo$,
	\item[2.] $\M_{A}$ is type II$_{1}$ (and the marginal $\omega_{A}$ is the unique normal tracial state) if and only if $$\sum_{k=1}^\oo\sum_{r=1}^{d_{k}}\abs[\big]{d_{k}^{-\frac{1}{2}}-(\lambda^{(k)}_{r})^{\frac{1}{2}}}^{2}<\oo.$$
	\item[3.] If $\lambda^{(k)}_{1}\ge\delta$ for some $\delta>0$ and all $k\in\NN$, then $\M_{A}$ is type III if and only if 
    $$
        \sum_{k=1}^\oo\sum_{r=1}^{d_{k}}\lambda^{(k)}_{r}\inf\set[\Big]{\abs[\Big]{\tfrac{\lambda^{(k)}_{1}}{\lambda^{(k)}_{r}}-1},C}=\oo \quad \text{for some $C>0$}.
    $$ 
\end{itemize}
A factor such as $\M_{A}$ is called an ITPFI (infinite tensor product of finite Type I) factor, and these factors exhibit all the possible types, i.e.\ I$_1$, I$_2$, \dots, I$_{\oo}$, II$_{1}$, II$_{\oo}$ and III$_{\alpha}$ for $\alpha\in[0,1]$, which are unique apart from the III$_{0}$ case \cite[Thm.~3.9]{araki1968classification}.

Thus, we conclude that the above construction leads to tame bipartite pure states if and only if the local bipartite states $\omega_{k}$ on the $C^*$-subalgebras $\MM_{d_{-k}}\ox\MM_{d_{k}}$ get closer and closer to product states in the sense of 1., while the other possible situations covered by 2.\ and 3.\ yield wild bipartite pure states. 
In particular, if the marginals of the local bipartite states $\omega_{k}$ converge sufficiently fast to the maximally mixed state, the factors $\M_{A}$ and $\M_{B}$ are given by the unique hyperfinite II$_{1}$, but the marginals $\omega_{A}$ and $\omega_{B}$ will not coincide with the unique tracial state $\ntr$.

For infinite tensor-product states, it follows from \cref{thm:SR_is_multiplicative} that $\SR(\omega)=\prod_{k=1}^\oo d_k$ because $\SR(\omega_k)=d_k$ by construction. 
This will only be finite in trivial cases. 

If we consider the situation of an infinite two-sided spin chain $\A = \ox_{j\in\ZZ}\MM_{2}$ such that all corresponding left and right spin pairs share the same entangled state $\omega_{k} = \omega_{\alpha}$ with vector $\Omega_{\alpha}= (1+\alpha)^{-1/2}(\ket0\ox\ket 0+\sqrt\alpha\ket1\ox\ket 1)$ for $\alpha\in(0,1)$, we will have $\M_{\A}=\R_{\alpha}$ -- the hyperfinite type III$_{\alpha}$ factor. 
Interestingly, the typical situation encountered in QFT (see \cref{sec:qft}) is that of type III$_{1}$ von Neumann algebras which can be produced by choosing two alternating $\alpha_{1},\alpha_{2}\in(0,1)$ such that $\tfrac{\log\alpha_{1}}{\log\alpha_{2}}\notin\QQ$.

\subsection{Vacuum in QFT}
\label{sec:qft}

In the setting of algebraic quantum field theory (AQFT) \cite{haag1996lqp} an example of a wild bipartite state is provided by the vacuum state $\omega$ on the bipartite algebra $\A$ generated by the local von Neumann algebras $\M_{A} = \A(\O_{A})$ and $\M_{B}=\A(\O_{B})$ of two causally separated space-time open regions $\O_A$ and $\O_B$. Under general assumptions \cite{fredenhagen1985modular}, it is known that the algebras $\M_{A}$ and $\M_{B}$ are type III$_{1}$ factors, sharing the vacuum vector $\Omega = \Omega_{\omega}$ as a common cyclic and separating vector because of the Reeh-Schlieder property. If, in addition, the split property\footnote{The local net $\O\mapsto\A(\O)$ of von Neumann algebras satisfies the split property, if $\A(\O_A)\vee\A(\O_B)\cong\A(\O_A)\Bar\otimes\A(\O_B)$ for arbitrary pairs of spacelike-separated open regions with $\Bar\O_A\cap\Bar\O_B=\emptyset$. Thr split property for quantum field theories is connected to but different from the split property for bipartite states in \cref{def:haag}.} is assumed these algebras will also be hyperfinite \cite{haag1996lqp}.

It follows that the vacuum $\omega$ is a pure bipartite state if and only if the inclusion of factors $\M_{A}\subset\M_{B}'$ is irreducible. This, for example, may happen for the algebras $\M_{A}$ and $\M_{B}$ of the right and left standard wedges $\O_{A}=\W_{\textup{R}} = \set{x\in\RR^{4} \given x_\mu x^\mu<0,\, |x_0|<x_1}$ and $\O_{B}=\W_{\textup{L}} = \set{x\in\RR^{4} \given x_\mu x^\mu<0, |x_0|<-x_{1}}$. Assuming that the quantum field theory satisfies wedge duality \cite{bisognano1976duality, buchholz1995scaling1}, we know that $\A_{A}$ and $\A_{B}$ satisfy Haag-duality and are in standard form. In particular, the weak closure of $\A$ in the vacuum representation $\H$ is $\B(\H)$, and, thus, $\omega$ is a pure state on $\A$. 

We point out that no contradiction with \cref{thm:pure_tame_states} arises if the quantum field theory model under consideration satisfies the split property \cite{buchholz1974product, Doplicher1984}, i.e.\ the von Neumann algebra $\M_{A}\vee\M_{B}\cong\M_{A}\Bar\ox\M_{B}$, for $\overline{\O_{A}}\cap\Bar\O_{B}=\emptyset$, admits normal product states extending $\omega_{A}$ and $\omega_{B}$ because $\omega$ will not be a pure bipartite state in this case.
\\

Finally, we include the result that no proper commuting operator correlations (see \cref{sec:Tsirelsons_prob}) between two spacelike separated systems can occur in a quantum field theory with hyperfinite local von Neumann algebras. A von Neumann algebra $\M$ is hyperfinite if it contains an increasing net of finite-dimensional *-algebras $\M_{\gamma}$ such that $\bigcup_{\gamma}\M_{\gamma}$ is $\sigma$-weakly dense.

\begin{prop}\label{thm:qft}
    Let $\O\mapsto \A(\O)$ be a local net of hyperfinite von Neumann algebras and let $\omega$ be a locally normal state\footnote{A state on the quasi-local algebra is called locally normal if its restriction to any of the local von Neumann algebras $\A(\O)$ is a normal state.} on the quasi-local algebra $\A = \Bar{\bigcup \A(\O)}$.

    Let $\O_A$ and $\O_B$ be two spacelike separated regions.
    If Alice and Bob choose their POVMs from their local observable algebras $\A(\O_A)$ and $\A(\O_B)$, only correlation functions $p(\alpha,\beta|i,j)$ of the correlation body $\C_{qa}$ can be obtained, i.e.\ can be realized approximately in finite-dimensions.
\end{prop}

Hyperfiniteness is a property that is expected to hold for general QFTs.
As stated above, hyperfiniteness follows if the local net $\O\mapsto\A(\O)$ satisfies the split property.
The above was already observed in \cite{scholz2008tsirelson}, but no explicit proof was given.
Since the details turn out to be slightly more involved than expected, we include a proof here:

\begin{proof}
    Let $\{M^j_{i,\alpha}\}_{\alpha=1}^k \subset \A(\O_j)$, $j=A,B$, be POVMs where $i=1,\ldots, n$ for two integers $n$ and $k$.
    We have to show that it is possible to approximate the correlation function 
    \begin{equation}
        p(\alpha,\beta|i,j)=\omega(M^A_{i,\alpha}M^B_{j,\beta}), \qquad \alpha,\beta\in\{1,\ldots,k\},\ \ i,j\in\{1,\ldots,n\}
    \end{equation}
    by correlation functions in $\C_q$ (see \cref{sec:Tsirelsons_prob}).

    First, note that we may assume that $\A$ and, hence, all $\A(\O)$ act on a Hilbert space $\H$ on which $\omega$ is implemented by a vector $\Omega\in\H$.
    Indeed if this is not the case, we consider the GNS representation $(\pi,\H,\Omega)$ of state $\omega$ on $\A$.
    Local normality has the consequence that the induced representations $\pi:\A(\O)\to\B(\H)$ are normal, implying that the von Neumann algebra $\pi(\A(\O))$ is hyperfinite for all $\O$.
    
    By the assumption of hyperfiniteness, there are directed sets $\Gamma_{j}$ and increasing nets of finite-dimensional subalgebras $\M_{\gamma_{j}}\subset\A(\O_j)$, $\gamma_{j}\in\Gamma_{j}$, such that $\bigcup_{\gamma_{j}}\M_{\gamma_{j}}\subset\A(\O_j)$ are $\sigma$-weakly dense (and, thus, also strong*-dense by Kaplansky's density theorem \cite[Thm.~II.4.8]{takesaki1}) for $j=A,B$.
    Pick nets of operators $F_{i,\alpha;\gamma_{j}}\in\M_{\gamma_{j}}$ with $\norm{F_{i,\alpha;\gamma_{j}}}\le 1$ which strong$^*$-converge to $\sqrt{M^j_{i,\alpha}}$.
    We define POVMs $M^j_{i,\alpha;\gamma_{j}}\approx (F_{i,\alpha;\gamma_{j}})^*F_{i,\alpha;\gamma_{j}}$ where "$\approx$" means that we allow for an error that strong$^*$-vanishes in the limit $\gamma_{j}\to\infty$ to ensure normalization, i.e.\ $\sum_\alpha M^j_{i,\alpha;\gamma_{j}}=\1$, for every $\gamma_{j}$.
    It can now be seen from the triangle inequality that $M^j_{i,\alpha;\gamma_{j}}\to M^j_{i,\alpha}$ in the strong operator topology.
    From this it follows that the resulting correlation functions $p_{\gamma_{A},\gamma_{B}}(\alpha,\beta|i,j)=\omega(M^A_{i,\alpha;\gamma_{A}}M^B_{j,\beta;\gamma_{B}})$ converge to $p(\alpha,\beta|i,j)$:
    \begin{align*}
        \abs{p(\alpha,\beta|i,j)-p_{\gamma_{A},\gamma_{B}}(\alpha,\beta|i,j)}
        & \le \abs{\ip{M^A_{i,\alpha}\Omega}{(M^B_{j,\beta}-M^B_{j,\beta;\gamma_{B}})\Omega}} + \abs{\ip{(M^A_{i,\alpha}-M^A_{i,\alpha;\gamma_{A}})\Omega}{M^B_{j,\beta;\gamma_{B}}\Omega}}\\
        &\le \norm{(M^A_{i,\alpha}-M^A_{i,\alpha;\gamma_{A}})\Omega} + \norm{(M^B_{j,\beta}-M^B_{j,\beta;\gamma_{B}})\Omega}\to0.
    \end{align*}
    Since, the algebras $\M_{\gamma_{A}}$ and $\M_{\gamma_{B}}$ are finite-dimensional and commute, we have $\M_{\gamma_{A}}\vee\M_{\gamma_{B}}\cong\M_{\gamma_{A}}\otimes\M_{\gamma_{B}}$ and $p_{\gamma_{A},\gamma_{B}} \in \C_q$.
\end{proof}

\section{Outlook on open problems}\label{sec:outlook}

We collect a list of open problems concerning the commuting operator framework.

\begin{labeledlist}{l}
    \item[\it Schmidt rank for mixed states.] 
        In finite-dimensional quantum mechanics, the Schmidt rank for a mixed state (often called the Schmidt number \cite{terhal2000schmidt}) $\rho$ is defined as $\min_{\{p_i,\psi_i\}} \max_i \SR(\psi_i)$ where the minimum is over all ensembles of pure states whose average state is $\rho$.
        Already for density operators on infinite-dimensional Hilbert spaces one needs to look at integrals over probability measures on the pure state space (infinite convex combinations are insufficient \cite{holevo2005separability}). 
        For non-separable $C^*$-algebras, $\A_A$ and $\A_B$, one runs into measure theoretic problems as the pure bipartite states are, in general, not a measurable subset of $\states_{max}$, the state space of the maximal $C^*$-tensor product.
        A good generalization of the Schmidt rank for mixed bipartite states should also generalize the hierarchy \eqref{eq:hierachy} to mixed states such that we have convex sets $\SR_{\le k}$ that interpolate between separable and tame bipartite states.
        Work on this is in progress.

    \item[\it A Schmidt rank-like invariant for wild states.]
        It would be nice to have a Schmidt rank-like quantity that can be finite even for wild bipartite pure states.
        In the tame case, a definition of the Schmidt rank as the dimension of the support projection of the reduced states acting on $\M_A$ or $\M_B$ makes sense and can readily be checked to be equivalent to the definitions we gave.
        This could be mimicked for wild states by using the dimension functions of the factors $\M_A$ and $\M_B$.
        However, the dimension function is only defined up to a constant factor in some cases. It is unclear where this normalization should come from and how one should choose compatible normalizations for $\M_A$ and $\M_B$.

    \item[\it Bell inequality violation of wild states.]
        The CHSH-Bell inequality asserts that in a local hidden variable theory, $\beta(\omega)\le 1$ holds for every bipartite state where 
        \begin{equation}\label{eq:beta}
            \beta(\omega) := \frac 12 \sup \omega(a_1b_1 + a_1b_2 + a_2b_2 -a_2b_2)
        \end{equation}
        with the supremum being over all observables $a_1,a_2\in\A_A$, $b_1,b_2\in\A_B$ so that $-1\le a_i\le 1$, $-1\le b_i \le 1$.
        It is not hard to see that $\beta(\omega)$ is a correlation invariant (in the sense of \cref{def:correlation_inv}) and that it may equivalently be computed using the von Neumann algebras $(\M_A,\M_B,\H_\omega,\Omega_\omega)$.
        Tsirelson's inequality $\beta(\omega)\le \sqrt2$ is valid in this setting, and we know that maximal violation, i.e.\ $\beta(\omega)=\sqrt2$, can occur for both tame and wild states. 
        An example for tame states is a maximally entangled two-qubit state, and an example in the class of wild states is given by the vacuum in certain QFTs where one has Haag-duality and the local factors are both hyperfinite type III$_1$ factors \cite{summers1987}.
        There are many open questions. For example: Do all wild pure bipartite states violate Bell's inequality?
        
    \item[\it Role of Haag-duality for correlations.]
        A bipartite pure state $\omega$ induces commuting factors $\M_j=\pi_\omega(\A_j)''$ on $\H_\omega$.
        We proved that $(\M_A,\M_B,\H_\omega)$ and, hence, also Haag-duality and the Jones-index $[\M_A\!:\!\M_B']$ are correlation invariants. What do these mean in terms of correlations?
        Roughly speaking, the Jones index measures the deviation from Haag duality, which still leaves us with the question of what Haag-duality means.
        A partial result in this direction is \cref{thm:Haag-duality_order_int} which shows that Haag-duality implies that any positive linear functional $\varphi$ on Bob's system with $0\le\varphi\le\omega_B$ can be approximated by conditioned states $\omega^a=\omega(a(\placeholder))$, $a\in [0,1]_{\A_A}$.
        An open problem is whether the converse holds, i.e.\ if this approximation property implies Haag duality.
        The idea that a failure of Haag-duality is related to $\A_A$ or $\A_B$ being too small is also supported by the following observation:
        Start with an irreducible subfactor inclusion and pick a bipartite system $\A_A,\A_B\subset \A$ with bipartite state $\omega$ as in \cref{thm:pure_states_and_subfactors}.
        Then one can enforce Haag-duality by enlarging one (or both) of the algebras $\A_A$ or $\A_B$ for which the GNS representation $\pi_\omega$ is faithful.

    \item[\it Schmidt rank for multi-partite states.]
        Multi-partite entanglement is a complicated subject already in the finite-dimensional case.
        The main reason is that there is no Schmidt decomposition: it is not possible to write every vector $\Psi\in\H_{1}\ox\ldots\H_{n}$ as $\sum_\alpha\lambda_\alpha \Phi\up1_\alpha\ox\ldots\ox\Phi\up n_\alpha$.
        However, the definition of the Schmidt rank through minimal compressions (i.e.\ \cref{it:op_split} in \cref{thm:schmidt_rank}) has a straightforward generalization. It assigns to a multi-partite pure state $\omega$ (a state on a multi-partite algebra) the smallest number $k$ so that there are unital completely positive maps $C_j:\A_j\to \MM_k$ and a vector $\Psi\in(\CC^{k})^{\ox n}$ that emulate $\omega$, i.e.\ $\omega(a_1\ldots a_n)=\ip\Psi{C_1(a_1)\ox\ldots\ox C_n(a_n)\Psi}$.
\end{labeledlist}
\null

\noindent Finally, we comment on a \emph{von Neumann algebraic version} of our results.
In some contexts, it makes sense to describe physical systems by von Neumann algebras $\A_A$ and $\A_B$, which imposes a regularity condition on states: Usually, the "physical" states of such systems states are only the normal states, i.e.\ states which are $\sigma$-weakly continuous.
The standard way to approach bipartite states in this setting would be to ask for a normal state on the von Neumann-tensor product $\M_A\Bar\ox\M_B$.
With such states, one can, however, not get proper commuting operator correlations (see \cref{sec:Tsirelsons_prob}).
We believe that the von Neumann algebraic analog of commuting operator framework correlations are \emph{locally normal states}, i.e.\ positive states on the algebraic tensor product $\M_A\ox\M_B$ whose marginals $\omega(\placeholder\ox\1_B)$ and  $\omega(\1_A\ox\placeholder)$ are both normal.
For such states, the mapping $(a,b)\mapsto \omega(a\ox b)$ is separately $\sigma$-weakly continuous. 
The idea of locally normal states has been around in AQFT for many years, where locality refers to space-time locality, but we are not aware of an abstract treatment.
It would be interesting to study the concept of bipartite algebras in the von Neumann algebraic setting.

\subsection*{Acknowledgements}

We thank Martin Pl\'avala for pointing out the factorization definition of the Schmidt rank of finite-dimensional systems.
We thank Thomas Cope, Marco Fanizza, Niklas Galke, Vern Paulsen, and Henrik Wilming for helpful discussions.

\appendix
\section{The Radon-Nikodym theorem for completely positive maps}\label{sec:appendix}
\setcounter{equation}{0}
\numberwithin{equation}{section}

We use the Radon-Nikodym theorem for completely positive maps on multiple occasions.
It seems to us that this result was first observed by Arveson in \cite[Thm.~1.4.2]{Arveson1969} (see \cite{Belavkin1986} for a generalization).
If $\A$ and $\B$ are $C^*$-algebras and $S,T:\A\to\B$ are completely positive maps, then we write $S\le_{cp}T$ if $T-S$ is completely positive and we denote by $[0,T]_{cp}$ the convex set of completely positive maps $S\le_{cp}T$.

\begin{lem}\label{thm:dominated_channels}
    Let $T : \A \to \B(\H)$ be a unital completely positive map. Let $(\pi,\tilde \H, V)$ be its minimal Stinespring dilation.
    There is a bijection between the order interval $[0,T]_{cp}$ and the unit interval $[0,\1]_{\pi(\A)'}$ of the von Neumann algebra $\pi(\A)'$ given by 
    \begin{equation}\label{eq:SQ_bij}
        S \leftrightarrow Q \iff S = V^* Q \pi(\placeholder) V.
    \end{equation}
    This bijection is affine and monotone.
\begin{proof}
    Recall that minimality of the Stinespring dilation means that $[\pi(\A)V\H]=\tilde\H$. 
    We use the shorthand $\Psi_{a,\psi} = \pi(a)V\psi$ with $a\in\A$, $\psi\in\H$, so that $\pi(a)\Psi_{a,\psi}= \Psi_{ab,\psi}$.
    It is clear that every $Q \in\pi(\A)'$ determines a linear map that is completely positive and dominated by $T$ if $0\le Q\le \1$.
    The idea for the converse is to define $Q$ from $S$ as the unique bounded operator whose matrix elements with respect to $\pi(\A)V\H$ are 
    \[
        \ip{\Psi_{a,\psi}}{Q\Psi_{b,\phi}} = \ip\psi{S(a^*b)\phi}.
        \] 
    Complete positivity of $S$ readily implies that this quadratic form is positive semi-definite.
    Now let $S'\ge S$, then the quadratic forms are ordered as $Q'\ge Q$ because
    \begin{align*}
        \ip[\Big]{\sum_i\Psi_{a_i,\psi_i}}{Q\sum_i\Psi_{a_i,\psi_i}}
        &= \sum_{ij} \ip{\psi_i}{S(a_i^*a_j)\psi_j} \\
        &\le \sum_{ij} \ip{\psi_i}{S'(a_i^*a_j)\psi_j}
        =\ip{\sum_i\Psi_{a_i,\psi_i}}{Q'\sum_i\Psi_{a_i,\psi_i}}.
    \end{align*}
    With $S'=T$ we get $0\le Q\le \1$ (as quadratic forms and, hence, as bounded operators).
    We now prove that $Q\in \pi(\A)'$. By linearity it suffices to check this on $\pi(\A)V\H$:
    \[
        \ip{\Psi_{a,\psi}}{\pi(b)Q\Psi_{c,\phi}} =\ip\psi{S((b^*a)^*c)\phi} = \ip{\psi}{S(a^*(bc))\phi} = \ip{\Psi_{a,\psi}}{Q\pi(b)\Psi_{c,\phi}}.
    \] 
    We check that \eqref{eq:SQ_bij} holds:
    \[
        \ip\psi{S(a)\phi} = \ip{\Psi_{1,\psi}}{Q\Psi_{a,\phi}} =\ip{ V\psi}{Q\pi(a)V\phi} = \ip\psi{V^*Q\pi(a)V\phi}.
    \] 
\end{proof}
\end{lem}

Since states are unital completely positive maps $\A\to\B(\CC)$, one obtains the well-known Radon-Nikodym theorem for states as a special case:

\begin{cor}\label{thm:affine_bij}
    Let $\omega$ be a state on a $C^*$-algebra $\A$.
    There is a bijection between the order interval $[0,\omega]\subset\A^*$ and the unit interval $[0,\1]_{\pi_\omega(\A)'}$ of the von Neumann algebra $\pi_\omega(\A)'$ given by 
    \begin{equation}\label{eq:affine_bij}
    \nu \leftrightarrow Q \iff \nu = \ip{\Omega_\omega}{Q\pi_\omega(\placeholder)\Omega_\omega}.
    \end{equation}
    This bijection is affine and monotone.
\end{cor}

We will now show that this bijection is completely positive in both directions.
This is probably known but we were unable to locate it in the literature.
In fact, the span of $[0,\omega]$ naturally carries an operator system structure (with order unit $\omega$ and $^*$-operation and matrix order inherited from $\A^*$) so that the bijection \eqref{eq:affine_bij} extends to a unital complete order isomorphism (an isomorphism of operator systems).
We need this result in the context of finitely correlated states (see \cref{sec:spinchains}).

\begin{prop}\label{thm:ucoi}
    Let $\omega$ be a state on a $C^*$-algebra $\A$ and set $\V= \lin\,[0,\omega]\subset\A^*$.
    Equip $\V$ with the matrix order and $^*$-operation inherited from $\A^*$ 
    \begin{enumerate}[(1)]
        \item 
            $\omega$ is an Archimedean matrix order unit for $\V$ (see \cite[Ch.~13]{paulsen2002completely}). Therefore $\V$ is an operator system.
        \item 
            Denote by $\lambda:\pi_\omega(\A)'\to\V$ the linear extension of the bijection \eqref{eq:affine_bij}.
            Then $\lambda$ is a unital complete order isomorphism, i.e.\  $\lambda$ and $\lambda^{-1}$ are completely positive and $\lambda(\omega)=\1$.
    \end{enumerate}
\end{prop}

\begin{proof}
    We only need to prove the second item since it implies the first one (because the identity is an Archimedean matrix order unit for $\pi_\omega(\A)'$). 
    We note here that the first item can also be proved directly and that this proof still holds if $\A$ is only assumed to be an operator system (the proof via the second item then fails because there is no Radon-Nikodym theorem in that case).

    Let $\S$ be matrix ordered $^*$-vector space  (e.g.\ $\S=\CC$, $\pi_\omega(\A)'$ or $\A^*$). 
    We identify $\MM_n\ox\S$ with $\MM_n(\S)$ ($n\times n$-matrices with values in $\S$) via $\sum \ketbra ij \ox x_{ij}= [x_{ij}]$ so that $X=\sum \bra i X\ket j \ketbra ij$.
    We set $\M=\pi_\omega(\A)'$ and let $n\in\NN$ be arbitrary.
    We need to show that $\id_n\ox\lambda:\MM_n\ox\M\to\MM_n\ox\V$ and its inverse preserve positivity where $\MM_n\ox\V$ is obtained with the positive cone inherited from the canonical matrix order of $\A^*$.
    To understand positivity in $\MM_n\ox\V$, we use Choi's Theorem \cite[Thm.~6.1]{paulsen2002completely}, which states that a matrix $[\varphi_{ij}]\in\MM_n\ox\A^*$ is positive if and only if the linear map $\Phi:\A\to\MM_n$ given by $a\mapsto [\varphi_{ij}(a)]$ is completely positive.
    Note that the Choi-isomorphism is, in particular, a linear bijection between $\MM_n\ox\A^*$ and linear maps $\A\to\MM_n$.
    It relates the matrix $\1_n\ox\omega$ to the unital completely positive map $\omega\up n:=\omega(\placeholder)\1_n:\A\to\MM_n$.
    The minimal Stinespring dilation of $\omega\up n$ is $(\pi,\H,V)=(\1_n\ox\pi_\omega, \CC^n\ox\H_\omega, \1\ox\ket{\Omega_\omega})$, where $\ket{\Omega_\omega}$ is the linear operator $\CC\ni z\mapsto z\Omega_\omega\in\H_\omega$ and $\bra{\Omega_\omega}=\ket{\Omega_\omega}^*:\H_\omega\ni\Psi\mapsto\ip{\Omega_\omega}\Psi\in\CC$.
    Denote by $\Lambda:\pi(\A)'=\MM_n\ox\M \to \lin\,[0,\omega\up n]_{cp}$ the linear extension of the bijection from \cref{thm:dominated_channels} which satisfies $\Lambda(\1_n\ox\1)=\omega\up n$.
    It follows that $\Lambda$ and $\lambda\ox\id_n$ are equal up the Choi-isomorphism:
    \begin{align*}
        \Lambda([Q_{ij}](a) 
        &= \sum_{ij} (\1\ox\bra{\Omega_\omega})(\ketbra ij \ox Q_{ij}\pi(a))(\1\ox\ket{\Omega_\omega})&&\\
        &= \sum_{ij}\ip{\Omega_\omega}{Q_{ij}\pi_\omega(a)\Omega_\omega}\, \ketbra ij 
        = [\lambda(Q_{ij})(a)], &&[Q_{ij}]\in\MM_n\ox\M.
    \end{align*}
    Therefore, the claim follows from combining Choi's theorem \cite[Thm.~6.1]{paulsen2002completely} and \cref{thm:dominated_channels}: 
    \[
        [Q_{ij}]\ge0 \iff \Lambda([Q_{ij}]) \text{ is completely positive} \iff (\lambda\ox\id_n)([Q_{ij}])\ge0.
    \] 

\end{proof}

\printbibliography 
\end{document}